\documentclass[conference]{IEEEtran}
\ifCLASSINFOpdf
\else
\fi
\usepackage{cite}
\usepackage{amsthm,amsmath,amssymb,amsfonts}
\usepackage{algorithmic}
\usepackage{graphicx}
\usepackage{textcomp}
\usepackage{xcolor}
\usepackage{url}
\usepackage{booktabs}
\usepackage[hidelinks]{hyperref}
\usepackage{pifont}
\usepackage{listings}
\usepackage{multirow}
\usepackage{subcaption}
\usepackage{makecell}

\newtheorem{theorem}{Theorem}

\lstdefinelanguage{EVM}{
  morekeywords={PUSH, ADD, CALLDATASIZE, LT, JUMP, CALLDATALOAD},
  keywordstyle=\color{blue}\bfseries, 
  sensitive=false
}

\lstdefinelanguage{RTL}{
  morekeywords={V1, V2, V3, V4, V5},
  keywordstyle=\color{blue}\bfseries, 
  sensitive=false
}

\begin{document}
%
\title{SEASONED: Semantic-Enhanced Self-Counterfactual Explainable Detection of Adversarial Exploiter Contracts}

\author{}
	

%


\IEEEoverridecommandlockouts
\makeatletter\def\@IEEEpubidpullup{6.5\baselineskip}\makeatother
\IEEEpubid{\parbox{\columnwidth}{
		Network and Distributed System Security (NDSS) Symposium 2025\\
		24-28 February 2025, San Diego, CA, USA\\
		ISBN 979-8-9894372-8-3\\
		https://dx.doi.org/10.14722/ndss.2025.[23$|$24]xxxx\\
		www.ndss-symposium.org
}
\hspace{\columnsep}\makebox[\columnwidth]{}}

\maketitle

\begin{abstract}
Decentralized Finance (DeFi) attacks have resulted in significant losses, often orchestrated through Adversarial Exploiter Contracts (AECs) that exploit vulnerabilities in victim smart contracts. To proactively identify such threats, this paper targets the explainable detection of AECs.

Existing detection methods struggle to capture semantic dependencies and lack interpretability, limiting their effectiveness and leaving critical knowledge gaps in AEC analysis. To address these challenges, we introduce SEASONED, an effective, self-explanatory, and robust framework for AEC detection.

SEASONED extracts semantic information from contract bytecode to construct a semantic relation graph (SRG), and employs a self-counterfactual explainable detector (SCFED) to classify SRGs and generate explanations that highlight the core attack logic. SCFED further enhances robustness, generalizability, and data efficiency by extracting representative information from these explanations. Both theoretical analysis and experimental results demonstrate the effectiveness of SEASONED, which showcases outstanding detection performance, robustness, generalizability, and data efficiency learning ability. To support further research, we also release a new dataset of 359 AECs. \footnote{
  The data, code and explanations results are archived on:  \url{https://doi.org/10.5281/zenodo.16753745}.
}
\end{abstract}


%
\IEEEpeerreviewmaketitle

\section{Introduction}
Decentralized Finance (DeFi), powered by blockchain ecosystem (e.g., Ethereum), has revolutionized financial systems through trustless transactions and programmable assets. Yet this innovation is persistently undermined by security attacks, ranging from oracle manipulations to flash loan exploits, that steal funds and destabilize protocols. Among these threats, \textbf{Adversarial Exploiter Contracts (AECs)} \footnote{They are also called Adversarial Contract \cite{ren2024lookahead} or Attacker Contract \cite{yang2024uncover} in the literature.}  -- malicious smart contracts \textit{deliberately designed and deployed by attackers} to exploit vulnerabilities in \textit{other smart contracts} --  stand as a critical attack vector.  In 2024 alone, AECs drove \$2.15 billion of the \$3 billion total losses in DeFi attacks \cite{Vismaya2024}, cementing their role as one of the dominant attack mechanisms eroding user trust and ecosystem sustainability.

Existing vulnerability detection tools \cite{bose2022sailfish,choi2021smartian,stephens2021smartpulse} focus on identifying \textit{internal flaws within a smart contract}, which share a critical limitation: they cannot prevent attacks due to blockchain's immutable nature. Once deployed, contracts cannot be modified, leaving identified weaknesses permanently exposed to exploitation. This \textit{fundamental constraint} necessitates \textbf{Adversarial Exploiter Contract Detection (AECD)} \cite{forta2023, yang2024uncover, ren2024lookahead}, which targets not passive vulnerabilities but active attack instruments—attacker-deployed contracts designed to exploit vulnerabilities. By identifying AECs before execution, AECD enables proactive countermeasures where patching is impossible. AECD thus establishes \textit{a vital paradigm shift}: from protecting victims through vulnerabilities detection to disarming attackers by weaponized AECs identification, thereby safeguarding the immutable ecosystems.

However, AECD presents significant challenges because the source code of AECs is typically unavailable. This absence obstructs direct analysis of the attack logic and impedes the extraction of meaningful malicious code patterns, thereby limiting the development of reliable heuristics for distinguishing AECs from benign contracts. Some approaches attempt detection by extracting statistical features from compiled bytecode of smart contracts, which is publicly accessible and generated when smart contracts are deployed on the blockchain system. Forta \cite{forta2023} applies TF-IDF and logistic regression to opcode sequences, while LookAhead \cite{ren2024lookahead} and Skyeye \cite{wang2024skyeye} leverage contract characteristics like function counts and call patterns. However, these methods exhibit critical limitations:

\begin{itemize}
    \item \textbf{Lack of Semantic Features}: They rely on syntactic features of bytecode, such as opcode distribution and the number of internal function calls. However, they overlook potential interactions and relational dependencies between different modules, which convey semantic information about the code logic. These overlooked features are often crucial for identifying AECs.
    \item \textbf{Lack of Interpretability}: They lack the capability to pinpoint the code fragments within AECs that is directly related to adversarial behaviors, making it difficult to comprehend the operational logic of AECs and to provide actionable insights for defense.
    \item \textbf{Statistically-driven}: These methods rely on statistical features that require extensive labeled data, which is scarce and expensive in practice, or lack generalizability, potentially missing newly emerging AECs with statistical features different from those in the training data.
\end{itemize}

To address these limitations, we propose SEASONED, an \textit{end-to-end} detector for AECD, that directly takes bytecode of smart contracts and outputs detection results \textit{as well as} explanations. Specifically, SEASONED first constructs a \textbf{Semantic Relation Graph (SRG)} from contract bytecode, capturing critical \textit{control, data, and effect} relations among code instructions. This graph-based representation facilitates an in-depth structural analysis of potential AECs without necessitating access to the source code. Subsequently, the SRG is analyzed by our proposed  \textbf{Self-Counterfactual Explainable Detection (SCFED)} module, which \textit{concurrently} provides: (1) a binary classification result (benign or AEC), and (2) counterfactual explanations that pinpoint the contract components most indicative of attack activity. Significantly, SEASONED operates exclusively on compiled bytecode and can be used pre-deployment, allowing for early identification and explanation of threats before attackers can exploit vulnerabilities on-chain, thereby helping to prevent potential asset losses.

Unlike existing detectors that solely focus on detection accuracy, SEASONED achieves dual objectives: 
\begin{itemize}
    \item \textbf{Obj.1 Interpretability Enhancement:} SEASONED integrates accurate detection with counterfactual explanations. These explanations comprise two key subgraphs derived from the Semantic Relation Graph (SRG): a \textit{factual subgraph} highlighting attack-relevant components, and a \textit{counterfactual subgraph} isolating benign components. That is, SEASONED not only reveals the core attack parts of AECs, but also eliminates redundant and noisy elements. This capability enables domain experts to analyze core attack logic, thereby facilitating a deeper understanding of adversarial mechanisms and advancing beyond the limitations of black-box detection models.
    \item \textbf{Obj.2 Multidimensional Performance Enhancement:}
    In practical applications, AEC detectors must not only demonstrate high accuracy, but also achieve concurrent advancements across three key dimensions:
    i) \textbf{Generalizability} - to adapt to novel, unforeseen AECs emerging beyond the scope of training data in this rapidly evolving ecosystem;  ii) \textbf{Robustness} - to withstand sophisticated adversarial evasion techniques like obfuscation and perturbations deployed by attackers; and
    iii) \textbf{Data Efficiency} - to learn effectively despite scarce, expensive expert annotations required for contract labeling. SEASONED delivers this multidimensional improvement intrinsically. Its core mechanism—prioritizing attack-relevant patterns from the factual subgraph—directly fortifies the model against novel threats, adversarial attacks, and data scarcity.
\end{itemize}

Crucially, while post-hoc explainers \cite{gnnexplainer,pgexplainer,cf2} can achieve interpretability (Obj.1) through after-the-fact analysis, they inherently cannot improve detector performance (Obj.2) because their explanations are decoupled from the training process. In contrast, our SCFED integrates explanation generation and detector enhancement in a unified framework. Thus, SCFED can simultaneously produce counterfactual explanations and strengthen the detector's capabilities in robustness, generalizability, and data-efficiency learning.

To evaluate the performance of SEASONED, we construct a real-world dataset from five public sources~\cite{defillama,rekt,slowmist,neptune,chainsec}. The comprehensive experimental results  show that our method not only accurately identifies AECs but also: (1) offers interpretable counterfactual explanations that elucidate the underlying logic of AECs; (2) maintains strong robustness against various perturbations; (3) demonstrates superior generalizability to newly emerging AECs compared to baseline methods; and (4) achieves highly efficient learning, delivering comparable performance with just 50 training samples as models trained on the entire dataset of over 1,000 samples.

Our key contributions can be summarized as follows.
\begin{enumerate}
    \item \textbf{Novel Semantic Representation:} We introduce a new representation for smart contract bytecode, the Semantic Relation Graph (SRG), which captures critical semantic relational information.
    \item \textbf{Effective Self-Explainable Detector:} We propose an effective self-explainable detector that delivers robust, generalizable, and data-efficient learning, providing both detections and explanations simultaneously. To our knowledge, this is the first self-explainable detector designed for adversarial exploiter contracts detection.
\end{enumerate}

This paper is organized as follows: Section II introduces DeFi ecosystems and Adversarial Exploit Contracts. Section III formally defines the problem of interpretable AEC detection. Section IV presents our solution—a self-explainable detection system combining semantic relation graphs representation and a self-counterfactual explainable detector. Section V validates the framework's detection performance and explanation quality across a real-world dataset. Finally, Section VI concludes with broader impacts and future directions.

\section{Background}

\subsection{DeFi and Smart Contract}

Decentralized Finance (DeFi) is a blockchain-powered financial system that leverages smart contracts to provide open, transparent, and accessible financial services without the need for traditional intermediaries such as banks or brokers.
By utilizing self-executing smart contracts, DeFi platforms automate complex financial operations such as lending, borrowing, trading, and asset management. 

The deployment of smart contracts constitutes a critical process in blockchain ecosystems. A smart contract encapsulates predefined behavioral logic through its source code, which is compiled into bytecode.
Subsequently, a contract creation transaction containing this bytecode is submitted to the blockchain. The transaction then enters the pending transaction pool, a repository for unverified incoming transactions, wherein the contract bytecode becomes publicly accessible. The final deployment is achieved upon the transaction's successful incorporation into a new block. Critically, our method enables the identification of AECs pre-deployment, thereby facilitating effective countermeasures to enhance smart contract security.

\begin{figure}[t]
    \centering
    \includegraphics[width=\linewidth]{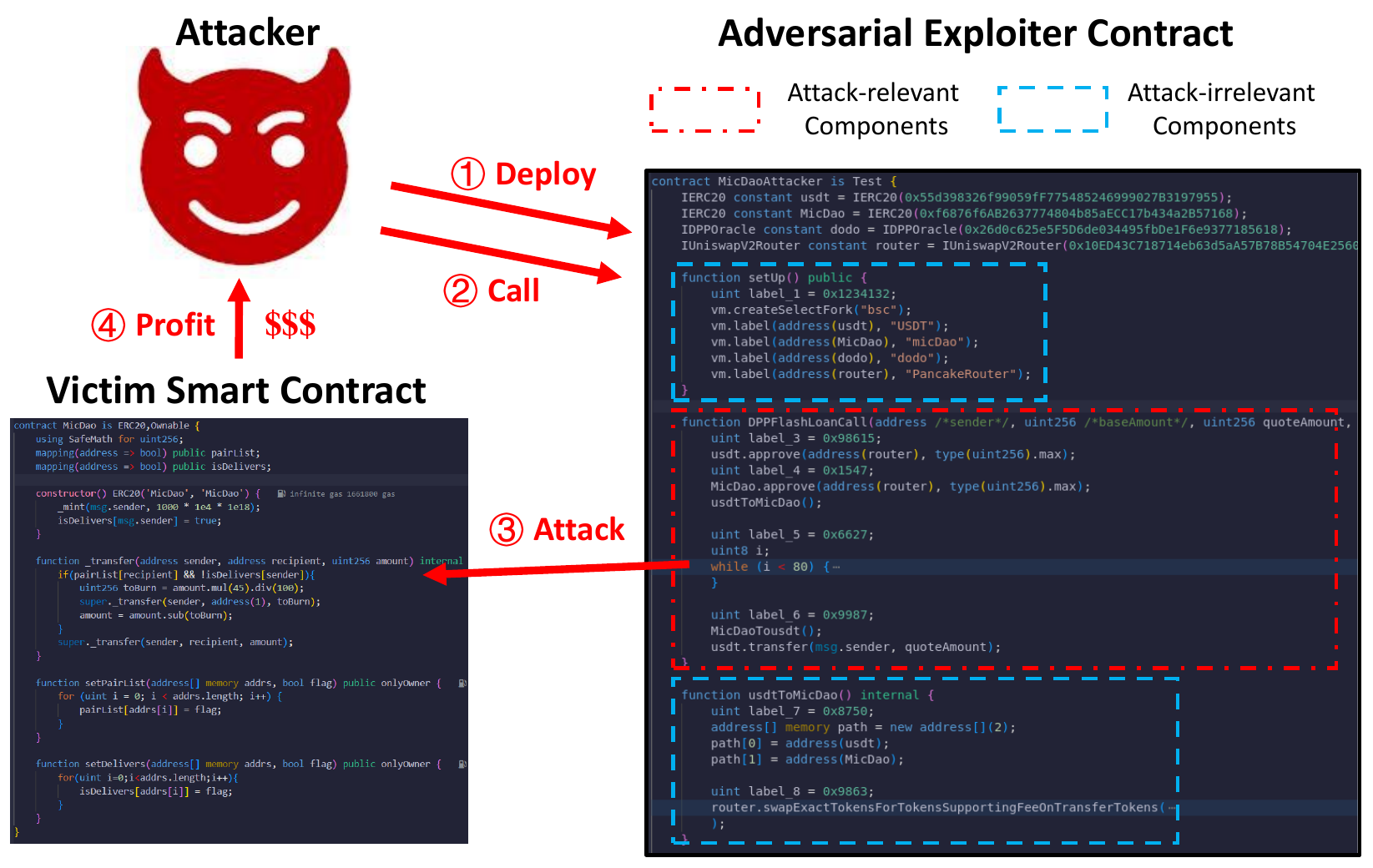}
    \caption{An attack incident involving AECs. The attacker first deploys an AEC and then sends a message to invoke it, which subsequently exploits and attacks vulnerabilities within the victim contract.}
    \label{fig:pattern}
\end{figure}

\subsection{Adversarial Exploiter Contracts (AECs)}

While smart contracts enable powerful and transparent financial services in DeFi, vulnerabilities in their implementation introduce significant security risks.
Attackers exploit these risks by deploying \textbf{Adversarial Exploiter Contracts (AECs)} – malicious contracts specifically designed to target victim contracts. These attacks typically involve AECs repeatedly invoking victim contracts to trigger malicious reentrancy or circumventing access controls to manipulate critical parameters like asset prices. A representative example of an AEC attack flow is illustrated in  Fig.~\ref{fig:pattern}. The attacker first deploys an AEC and subsequently invokes its entry-point function to exploit vulnerabilities in victim contracts for financial gain.

The severe threat posed by AECs demands urgent attention. A recent study \cite{OWASP2025} documented 149 attack incidents involving AECs, inflicting staggering losses totaling \$1.42 billion in 2024 alone. Crucially, the closed-source nature of AECs actively conceals their attack logic and operational mechanisms, creating fundamental and persistent gaps in current detection capabilities. This combination of immense financial impact and inherent obfuscation underscores the critical necessity for robust, effective methods to detect AECs proactively.


\begin{figure*}[t]
    \centering
    \includegraphics[scale=0.55]{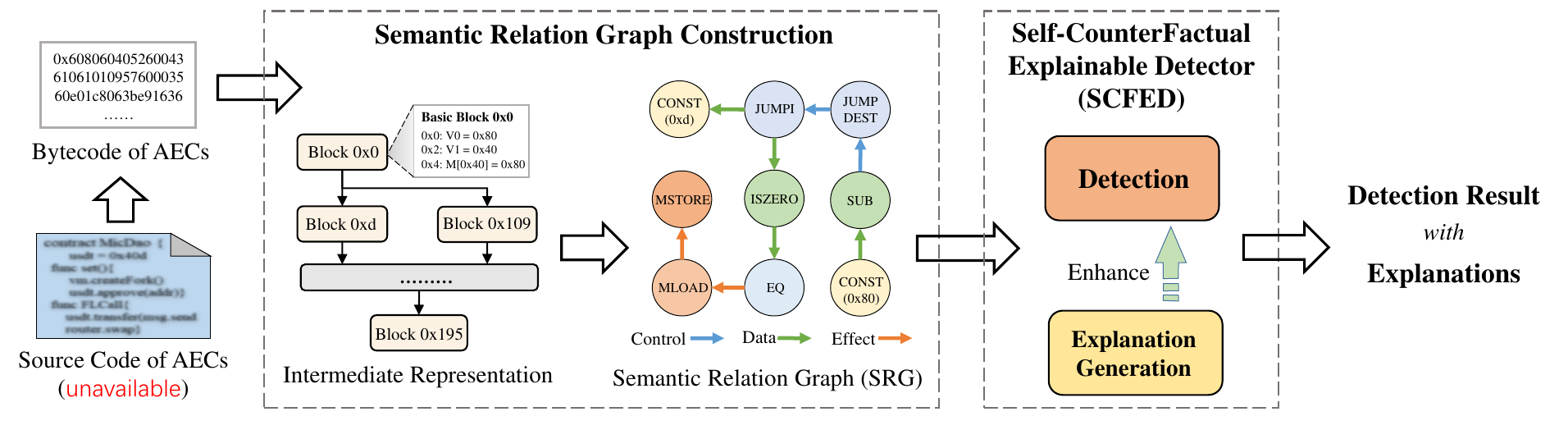}
    \caption{SEASONED consists of two main modules: Semantic Relation Graph Construction and the Self-Counterfactual Explainable Detector (SCFED). The first module builds a Semantic Relation Graph (SRG) from the bytecode of AECs, while SCFED delivers detection results (AEC or benign) along with explanations that highlight the core attack logic. }
    \label{fig:overview}
\end{figure*}

\section{Problem Statement}

\subsection{Threat Model}

\subsubsection{DeFi Attacks based on AECs}

We focus on DeFi attack incidents carried out by AECs, such as flash loan attacks and reentrancy attacks. A key characteristic of these attacks is that attackers do not interact directly with the victim; instead, they initiate a pre-deployed AEC to exploit the victim contract's vulnerabilities, making the attack more covert. Recent studies \cite{ren2024lookahead} indicate that attackers can initiate AECs through private transactions, thereby evading the detection of transaction-based methods \cite{zhang2020txspector, ferreira2021eye, xi2024pomabuster} and rendering such methods largely ineffective. It is worth noting that DeFi attacks not relying on AECs are beyond the scope of this paper.

\subsubsection{Attackers and Defender}
In this paper, attackers refer to users within the blockchain ecosystem that seek to gain profit through exploiting vulnerable smart contracts. The defenders, on the other hand, represent victims or regulatory authorities that maintain and secure the blockchain ecosystem.

\textbf{Attackers' Goal:} The attacker aims to exploit vulnerabilities in victim contracts by creating and deploying AECs to gain economic benefits.

\textbf{Attackers' Knowledge:} Attackers can access the source code and binary code of all contracts, including AECs and benign contracts. However, attackers do not know the SRG (Semantic Representation Graph) of these contracts, nor the defender’s defense strategies or implementation mechanisms.

\textbf{Attackers' Capability:} Attackers can modify the source code of AECs, but cannot directly control the corresponding binary code or SRG. This means that while the attacker can indirectly influence the binary and SRG by changing the AECs’ source code, this influence is not fully controllable.

\textbf{Defender's Goal:} During the smart contract pending period, the defender aims to distinguish between AECs and benign contracts using binary code analysis, and to provide explanations for detected AECs, clarifying the reasons for their classification and revealing the core attack logic of AECs.

\textbf{Defender's Knowledge:} The defender has access to the binary code and SRG of all contracts, as well as the source code of the victim and benign contracts, but cannot access the source code of AECs.

\textbf{Defender's Capability:} The defender cannot write or modify the source code, binary code, or SRG of any contract. This limitation means that commonly used purifying- or sanitation-based robust methods are not applicable. Therefore, the defender must rely on inherently robust detection methods.

\subsection{Adversarial Exploiter Contract Detection (AECD)}
We focus on detecting AECs before their deployment by distinguishing contracts based solely on their compiled bytecode. We propose an automated detection system that, for each input smart contract, (i) predicts whether it is an attacker contract and (ii) provides explanations to support its predictions.
Notably, we focus exclusively on attacks that require the involvement of AECs. DeFi attacks that do not require AECs \cite{zhou2024towards,chen2020phishing,kell2023forsage} are beyond our scope.

Specifically, let a smart contract bytecode be represented as \( C \). Given a dataset containing $N$ labeled smart contracts \( \{(C_i, y_i)\}_{i=1}^{N} \), where \( y_i \in \{0, 1\} \) indicates whether a contract is malicious (\( y_i = 1 \)) or benign (\( y_i = 0 \)), our goal is to propose a system $\mathcal{M}=\{\mathcal{G}, \mathcal{D}\}$, including a semantic relation graph constructor $\mathcal{G}$ and a self-explainable counterfactual detector $\mathcal{D}$. $\mathcal{G}$ constructs comprehensive representation \( G(V,E) \) for arbitrary bytecode \( C_i \), where $G$ is a semantic relation graph, $V$ and $E$ are node set and edge set respectively. $\mathcal{D}$ provides predictions $Y$ along with explanations ($S$, $R$) for \( G \), where $S$ is the factual subgraph and $R$ is the counterfactual subgraph:
\begin{align}
    (S,R,Y)=\mathcal{D}(G),\; G=\mathcal{G}(C_{i}), 1\leq i \leq N.
\end{align}

$S(V_{S}, E_{S})$ and $R(V_{R}, E_{R})$ fulfill:
\begin{align}\label{eq:sandr}
    & E_{S}\cup E_{R}=E, \; E_{S}\cap E_{R}=\emptyset, \; (E_{S}\subseteq E, \; E_{R}\subseteq E), \\
    & V_{S}\cup V_{R}=V, \; V_{S}\cap V_{R}\neq \emptyset, \; (V_{S}\subseteq V, V_{R}\subseteq V).
\end{align}

This means that all edges in $G$ are partitioned into two disjoint sets: one forming $S$ and the other forming $R$. $S$ contains the part that represents the core logic, while $R$ consists of nodes and edges that are considered irrelevant. Notably, the nodes contained in $S$ and $R$ may overlap, since edges connected to the same node can belong to different sets.




\section{Methodology}
In this section, we introduce SEASONED, our proposed self-explainable detection system for identifying Adversarial Exploiter Contracts. As shown in Fig.~\ref{fig:overview}, the system architecture comprises two integrated modules: a \textbf{data processing module} that constructs Semantic Relation Graphs (SRGs) from smart contract bytecode, and a \textbf{graph learning module} that performs joint prediction and explanation generation. 

Specifically, the data processing module disassembles contract bytecode into opcode sequences and builds SRGs to capture semantic relationships within each contract.
The graph learning module, guided by the Graph Information Bottleneck (GIB) \cite{GIB}, predicts whether an SRG is AEC or benign and generates counterfactual explanations. These explanations serve a dual purpose: they identify critical attack-relevant components while intrinsically enhancing detection performance. By learning representative information through GIB, the module achieves three interconnected advantages: it i) filters adversarial noise and decoy patterns to strengthen robustness against evasion techniques; ii) captures implementation-agnostic attack logic to improve generalizability across novel AEC variants; and iii) concentrates learning on fundamental malicious patterns to boost data efficiency significantly.


Together, these modules form a cohesive self-explainable detector where interpretability mechanisms actively contribute to functional improvements. The following subsections will elaborate on the technical implementation of both components.



\subsection{Semantic Relation Graph Representation}

\subsubsection{Data Extraction}
Smart contracts are written in high-level programming languages like Solidity and are subsequently compiled into low-level bytecode for automatic execution on the Ethereum Virtual Machine (EVM).
To conceal malicious intentions, AECs often do not disclose their source code publicly.
Accordingly, our analysis focuses on the bytecode.
We build a scanner to identify contract creation transactions in the pending pool.
In a creation transaction, the \texttt{To} field is empty (\texttt{null} or \texttt{0x0}) while a message call transaction has a \texttt{To} field that points to an existing contract or account.
Then, we extract the \texttt{input} field of these transactions, which contains the creation bytecode of contracts.
The extracted bytecode is passed on for further analysis.

\subsubsection{Intermediate Representation}
In this stage, we build the intermediate representation of contract bytecode.
We first convert the contract bytecode into EVM opcodes, as is shown in Fig.~\ref{fig:srg}(a).
EVM opcodes operate in a stack-based manner where operands are pushed onto the stack and operations are performed by popping these operands off the stack.
The EVM opcode representation makes it difficult to discern the relationships across opcodes.
For example, given EVM opcodes as shown in Listing \ref{lst:opcode}, the operands for each \texttt{ADD} opcode are not explicitly linked to their results, requiring careful tracking of the stack states to understand which values are being added.

\noindent
\begin{minipage}[t]{0.45\linewidth}
\lstset{language=EVM}
\begin{lstlisting}[caption={EVM Opcodes}, label={lst:opcode}]
02: PUSH 0x80;
04: PUSH 0x40;
08: ADD;
09: PUSH 0x20;
0b: ADD
0c: JUMP
\end{lstlisting}
\end{minipage}
\hfill
\begin{minipage}[t]{0.45\linewidth}
\lstset{language=RTL}
\begin{lstlisting}[caption={Register Transfer Language}, label={lst:rtl}]
V1 = 0x70; 
V2 = 0x40;
V3 = ADD V1 V2;
V4 = 0x20;
V5 = ADD V3 V4
JUMP V5
\end{lstlisting}
\end{minipage}

To address this issue, we transform opcodes into Register Transfer Language (RTL).
In the RTL format, registers are explicitly named and each register can only be assigned once. 
In the listing \ref{lst:rtl}, each step is explicitly defined, showing the exact data flow through the instructions.
Based on the RTL form, we proceed to construct the Control Flow Graph (CFG).
The process begins with the identification of basic blocks—groups of instructions that are executed sequentially without any jumps or branches.
Each basic block is then represented as a node in the CFG. 
Subsequently, we analyze the control flow opcodes, such as \texttt{JUMP} and \texttt{JUMPI}, which typically appear at the end of each basic block to determine their target destinations.
For instance, in Fig.~\ref{fig:srg}(b), \texttt{Block 0x109} is a destination of the \texttt{JUMPI} instruction (\texttt{0xc:JUMPI 0x109 V4} ) at the end of \texttt{Block 0x0}.
These control flow instructions define the edges of control flow graph.






\begin{figure*}[ht]
    \centering
    \includegraphics[width=\textwidth]{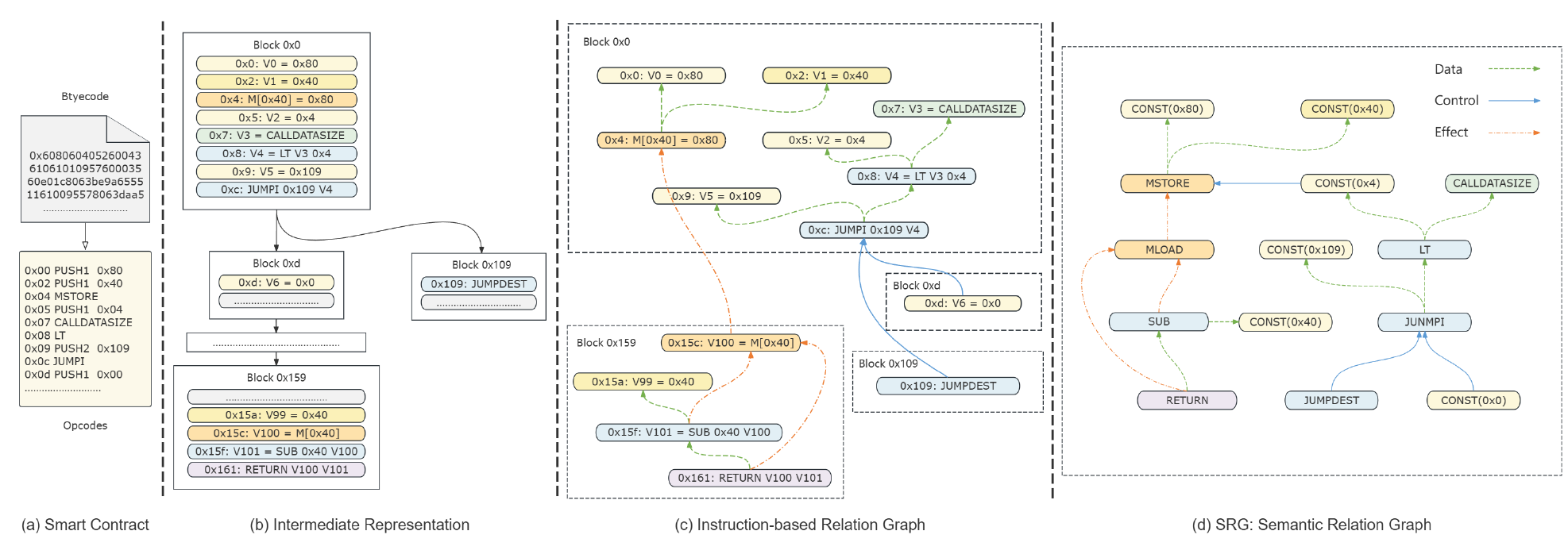}
    \caption{An illustrative example demonstrates the conversion process from contract bytecode to semantic relation graph. 
    (a) Initially, smart contracts are represented as bytecode, which is subsequently converted into opcode sequences. 
    (b) Contract opcodes are then transformed into register transfer language, and a control flow graph is constructed. 
    (c) Based on IR, control, data, and effect relations are established to form an instruction-based relation graph. 
    (d) Finally, the instruction-based relation graph is refined into a semantic relation graph.}
    \label{fig:srg}
\end{figure*}

\subsubsection{Semantic Relation Graph}

In this stage, we construct semantic relation graph based on the previously generated intermediate representation.
We use the AEC from a real-world attack, which targeted Inverse Finance \cite{halborn2022inverse}, as an illustrative example to present the detailed construction process. 

Observing the control flow graph in Fig.~\ref{fig:srg}(b), we can notice the implicit dependency relations between instructions.
To be specific, the execution order of \texttt{0xd:V6=0x0} and \texttt{0xc:JUMPI 0x109 V4}can not be changed according to the control flow path, while swapping instructions \texttt{0x0:V0=0x80} and \texttt{0x2:V1=0x40} do not affect the program result.
Nevertheless, instruction \texttt{0x4:M[0x40]=0x80} should execute after \texttt{0x0:V0=0x80} and \texttt{0x2:V1=0x40} as there is data dependency, and it should execute before \texttt{0x15c:V100=M[0x40]} as they both make a side effect in the same memory slot.

Based on this observation, we formulate three relations as edges between instruction nodes: control flow relation, data flow relation, and effect flow relation, to capture semantic relationships.
To obtain concise and efficient node representations, we simplify instructions into semantic opcodes.
Each instruction is condensed into a specific operator that encapsulates its primary semantics.
For operations like \texttt{V1=0x40}, we introduce a new semantic opcode \texttt{CONST} to indicate the assignment of a literal value. 
Stack-related operations such as \texttt{POP}, \texttt{PUSH1}-\texttt{PHSU32}, \texttt{SWAP1}-\texttt{SWAP16}, and \texttt{DUP1}-\texttt{DUP16} are removed as they do not imply semantic information.


\textbf{Control Flow Relation:} 
The control flow relationship captures the control-related semantics, wherein instructions or basic blocks are executed according to control structures such as conditionals.
The construction of the control flow relation edge is inherited from CFG.
Control flow edges between basic blocks are lifted based on specific branch-related semantic opcodes such as \texttt{JUMP}, \texttt{JUMPI}, and \texttt{JUMPDEST}. 
Let \( B = \{b_1, b_2, \dots, b_n\} \) represent the set of basic blocks, where each block contains a sequence of low-level opcodes.
Considering there exists a jumping operation from $b_i$ to $b_j$, the control relation edges \( E_C \) are constructed by connecting the branch opcodes connecting two basic blocks:
\begin{equation}
    E_C = \{(in_j, in_i) \mid in_j \in b_j, in_i \in b_i\},
\end{equation}

where $in_j$ represents a \texttt{JUMPDEST} instruction in block $b_j$, and $in_i$ is either a \texttt{JUMP} or a \texttt{JUMPI} instruction in block $b_i$.
$E_C$ is a directed edge from block $b_j$ to block $b_i$.



\textbf{Data Flow Relation:} Data flow relation reveals the def-use relationship, which refers to the connection between where a variable is defined (assigned a value) and where it is subsequently used (referenced). 
The construction of the def-use relation edge entails identifying and mapping the definitions and uses of variables within a program.
Since the contract bytecode has been converted to RTL format, the variable name and the node where it is assigned are extracted and stored as key-value pairs in the definition sites.

For each argument of an operation, if the operation originates from a literal value, the semantic opcode \texttt{CONST} is recorded.
For each argument of an operation, if the operation is not a direct constant assignment (i.e., its value originates from a literal value), the name and location of the variable used are recorded, and the corresponding entry in the use site mapping is updated.
Data dependencies are established whenever a variable is used and the variable already appears in the definition sites.
This dependency is represented as a def-use edge within the graph.
Let \( \text{Def}(v) \) and \( \text{Use}(v) \) denote the sets of definition and use sites of variable \( v \). The data flow relation edges \( E_D \) is constructed as:
\begin{equation}
    E_D = \{(in_{\text{u}}, in_{\text{d}}) \mid in_{\text{u}} \in \text{Use}(v), in_{\text{d}} \in \text{Def}(v) \},
\end{equation}

where \( in_{\text{u}}, in_{\text{d}} \) indicate instructions that use and define \( v \).

\textbf{Effect Flow Relation:}
The effect flow relation is proposed to capture the instruction that has a side effect on EVM storage or memory.
This relation is crucial as it reveals the constriction of the relative execution order between instructions that use or modify identical storage or memory slots.
For storage/memory store operations (e.g., \texttt{SSTORE}, \texttt{MSTORE}), the memory slot being affected is identified, and the instruction is recorded as a definition site in a mapping indexed by the memory slot.
For storage/memory load operations (e.g., \texttt{SLOAD}, \texttt{MLOAD}), the following steps are performed: (1) If the left-hand side variable of the instruction has a corresponding entry in the use mapping, its associated use site is retrieved and stored, else a new entry is created and populated with the use site. (2) The definition instruction corresponding to the use site is identified, and an effect edge is added to the SRG. (3) If a memory slot used by the instruction corresponds to a previously recorded definition site, an effect edge is similarly added to the SRG, connecting the current instruction to the relevant definition operation. 
We define effect flow relation edges as:
\begin{equation}
    E_E = \{(in_{\text{useEffect}}, in_{\text{effect}}) \mid in_{\text{effect}} \in \text{Def}(s)\},
\end{equation}

where \( s \) \text{ is a memory/storage variable}, \( in_{\text{effect}} \) corresponds to instructions such as \texttt{SLOAD} or \texttt{SSTORE}.

\begin{figure}[!t]
    \centering
    \includegraphics[scale=0.5]{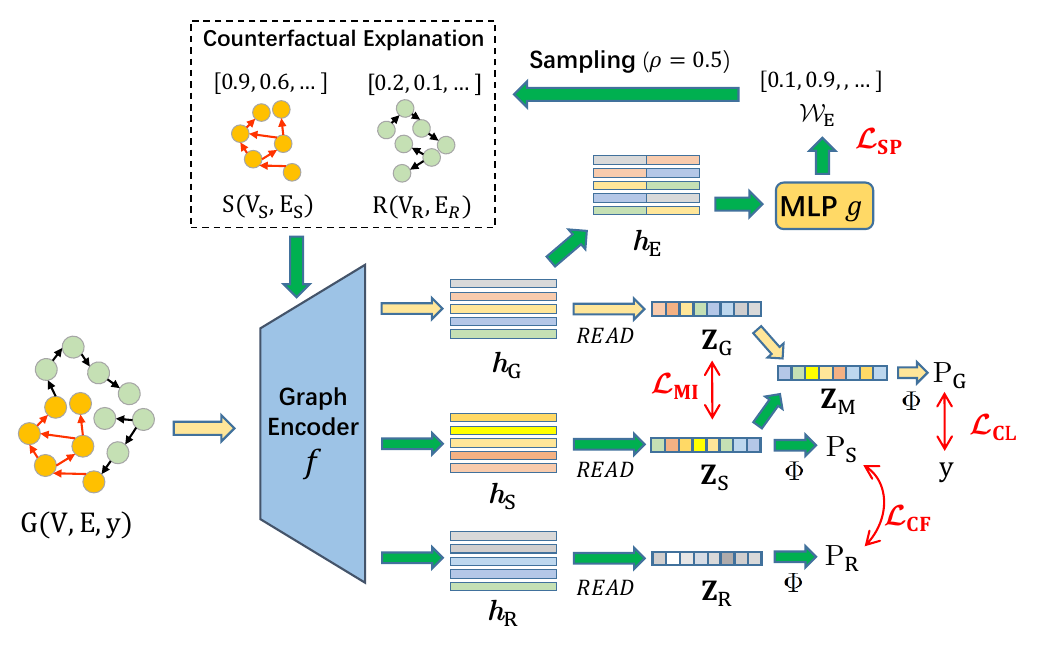}
    \caption{Framework of SCFED.}
    \vspace{-0.35cm}
    \label{fig:framework}
\end{figure}

\subsection{Detection with Self-CounterFactual Explainability}

To achieve both \textbf{Obj.1} and \textbf{Obj.2}, we propose Self-CounterFact Explainable Detector (SCFED). At a high level, given a SRG, SCFED  utilizes the self-explanatory technique to generate a factual subgraph and a counterfactual subgraph as the explanation (\textbf{Obj.1}). Furthermore, SCFED is inspired by GIB \cite{GIB}, which seeks to learn concise yet informative graph representations, and leverages this principle to extract representative embeddings from the factual subgraph, thereby enhancing detection performance (\textbf{Obj.2}).

\textbf{Overview.\ } As illustrated in Fig.~\ref{fig:framework}, SCFED integrates three core components: (1) a GNN encoder $f$ that generates node embeddings from the Semantic Relation Graph (SRG); (2) an MLP $g$ that computes edge weights from these embeddings, where higher-weight edges (indicating stronger predictive contribution) are assigned to the factual subgraph $S$ and lower-weight edges to the counterfactual subgraph $R$; and (3) a fully connected layer $\Phi$ that produces the graph prediction from aggregated node embeddings.  

Since the factual subgraph contains the core attack logic, we follow the approach of GIB and propose the loss function $\mathcal{L}_{\mathrm{MI}}$, which encourages the model $f$ to extract representative embeddings from $S$ and integrate them into the overall graph embeddings for label prediction.

We first detail SCFED's technical implementation and defer its theoretical foundations to the end of this section.




\subsubsection{Generation of Counterfactual Explanation} 
To achieve \textbf{Obj.1}, SCFED generates the counterfactual(cf) explanation $S(V_{S},E_{S})$ and $R(V_{R},E_{R})$ of the input graph $G(V,E,y)$. Essentially, this process can be formulated as an edge assignment problem. Inspired by the learnable edge dropping strategies proposed in~\cite{adgcl, pgexplainer}, we leverage a learnable edge assignment strategy to implement that. 

Specifically, SCFED utilizes learnable edge weights $\mathcal{W}_{E}\in\mathbb{R}^{|E|}$ to determine whether each edge should be assigned to subgraph $S$ or $R$:
\begin{equation}
    \mathcal{W}_{E} = g(h_{E}),\; h_{E} = T\cdot h_{N},\; h_{N} = f(G),
\end{equation}

where $g: \mathcal{R}^{d}\rightarrow\mathcal{R}$ is a MLP that learns edge weights from the $d-$dimansional edge embeddings $h_{E}\in\mathbb{R}^{|E|\times d}$, which is calculated from node embeddings $h_{N}\in\mathbb{R}^{|V|\times d}$. $T\in\mathbb{R}^{|E|\times |V|}$ is an incident matrix where $T_{ij}=0.5$ and $T_{ji}=0.5$ if $j-$th node is source or destination of $i-$th edge, otherwise $T_{ij}=0$ and $T_{ji}=0$. This ensures that the embedding of each edge is the average of the embeddings of its two endpoint nodes. $T$ is precomputed during preprocessing.

To ensure learnable explanation generation during the training, we apply the Gumbel-Max reparametrization trick \cite{gumblemax1,gumblemax2} to calculate a random variable $p_{e}\sim Bernoulli(w_{e})$ for $\forall e(u,v)\in E$ based on its edge weight $w_{e}\in \mathcal{W}_{E}$:
\begin{equation}
    p_{e}=Sigmoid((\log\delta-\log(1-\delta)+w_{e})/\tau),
\end{equation}

where $\delta\sim Uniform(0,1)$ and $\tau$ is a hyper-parameter. SCFED then samples edges with $p_{e}$ greater than the threshold $\rho$ to form the factual subgraph $S(V_{S}, E_{S})$, and other edges to form $R(V_{R}, E_{R})$:
\begin{equation}
    \begin{cases}
        e\in E_{S},\; (u,v)\in V_{S}, \quad \text{if}\;\; p_{e} > \rho, \\
        e\in E_{R},\; (u,v)\in V_{R}, \quad \;\,\text{otherwise}.
    \end{cases}
\end{equation}

Note $E_{S}\cup E_{R}=E$ and $E_{S}\cap E_{R}=\emptyset$, while $V_{S}\cup V_{R}=V$ but $V_{S}\cap V_{R}\neq \emptyset$.

To force SCFED to generate a concise explanation, $g$ is optimized by the sparsity loss $\mathcal{L}_{\mathrm{SP}}$:
\begin{equation}
    \mathcal{L}_{\mathrm{SP}} = ||\mathcal{W}_{E}||_{1}.
\end{equation}

$\mathcal{L}_{\mathrm{SP}}$ ensures that most elements in $\mathcal{W}_{E}$ are small, thereby resulting in fewer edges contained in $S$.

Following the acquisition of $S$ and $R$, we feed them into $f$ to obtain their graph-level embeddings $\mathrm{Z}_{S}, \mathrm{Z}_{R}\in\mathbb{R}^{d}$ and probabilities  $\mathrm{P}_{S},\mathrm{P}_{R}\in [0, 1]$, representing the likelihood that $S$ and $R$ are malicious, respectively.
\begin{align}\label{eq:subembedding}
    &\mathrm{P}_{S} = \Phi(\mathrm{Z}_{S}),\;\mathrm{Z}_{S}=READ(h_{S}),\; h_{S}=f(S), \\
    &\mathrm{P}_{R} =\Phi(\mathrm{Z}_{R}),\;\mathrm{Z}_{R}=READ(h_{R}),\; h_{R}=f(R),
\end{align}

where $\Phi: \mathcal{R}^{d}\rightarrow\mathcal{R}$ is a fully connected layer and $READ$ is the READOUT function of GNNs that aggregates node embeddings to derive a graph embedding. 

According to the property of counterfactual explanation, the counterfactual loss $\mathcal{L}_{\mathrm{CF}}$ is:
\begin{equation}\label{eq:loss_cf}
    \mathcal{L}_{\mathrm{CF}}=(1-y)\log(1+\mathrm{P}_{S}+\mathrm{P}_{R})-y\log\frac{(1+\mathrm{P}_{S}-\mathrm{P}_{R})}{2}.
\end{equation}

The first term of $\mathcal{L}_{\mathrm{CF}}$ ensures the predictions of both $S$ and $R$ are benign when the ground truth is benign (i.e., $y=0$). The second term ensures the prediction of $S$ is malicious while the prediction of $R$ is benign when $y=1$. 

For an input graph $G$ representing an AEC, $\mathcal{L}_{\mathrm{CF}}$ enforces the predictions of $S$ and $R$ are malicious and benign respectively, while $\mathcal{L}_{\mathrm{SP}}$ simultaneously constrains $S$ to maintain the minimal size. This optimization yields SCFED learns a compact yet discriminative subgraph $S$ that exclusively contains the core attack components responsible for G's malicious behavior. Simultaneously, attack-irrelevant graph elements (nodes and edges) are partitioned into $R$, thus filtering out benign elements that may otherwise obscure detection.

In summary, SCFED generates concise counterfactual explanations by minimizing $\mathcal{L}_{\mathrm{SP}}$ and $\mathcal{L}_{\mathrm{CF}}$.

\subsubsection{Explanation Enhanced Detection}
Since AECs in the real-world contain both attack-relevant components and attack-irrelevant components, the graph-level embedding $Z_{G}$ contains both attack-related and irrelevant information. In contrast, $S$ serving as a factual subgraph primarily captures the most attack-relevant parts of $G$. SCFED follows GIB theory to extract representative embeddings from $S$ to enhance the detection, thus achieving \textbf{Obj.2}. 

Specifically, SCFED merges $Z_{G}$ and $Z_{S}$ into $Z_{M}\in\mathbb{R}^{d}$ and minimizes the classification loss $\mathcal{L}_{\mathrm{CL}}$: 
\begin{align}\label{eq:loss_cl}
    & \mathcal{L}_{\mathrm{CL}}(\mathrm{P}_{G},y),\quad \mathrm{P}_{G}=\Phi(Z_{M}),\\
    Z_{M} &= \psi(Z_{G},Z_{S}),\; Z_{G}=READ(h_{G}),
\end{align}

where $\psi: \mathcal{R}^{d}\times\mathcal{R}^{d}\rightarrow\mathcal{R}^{d}$ is a linear mapping function (e.g., mean). $\mathcal{L}_{\mathrm{CL}}$ is a commonly used classification loss, such as the Cross-Entropy.

Eq.~\eqref{eq:loss_cl} indicates that SCFED leverages both graph-level information and subgraph-level information for detection by merging $Z_{G}$ and $Z_{S}$ into $Z_{M}$. 

To improve the representativeness of $Z_{S}$, SCFED follows GIB to minimize the mutual information shared by $Z_{G}$ and $Z_{S}$ via minimizing the mutual information loss $\mathcal{L}_{\mathrm{MI}}$:
\begin{align}\label{eq:los_mi}
    & \mathcal{L}_{\mathrm{MI}}= t^{*}(Z_{\mathrm{S}},Z_{\mathrm{G}}),\\
    \mathrm{s.t.}\; & t^{*}=\arg\min_{t} \;\mathcal{L}_{\mathrm{mine}}(t(Z_{\mathrm{S}},Z_{\mathrm{G}})).
\end{align}

Since the mutual information between $Z_{G}$ and $Z_{S}$ can not be calculated directly, we apply a MINE\cite{MINE} estimator $t: \mathcal{R}^{d}\times\mathcal{R}^{d}\rightarrow\mathcal{R}^{+}$ to approximate it. $t^{*}$ is a trained estimator by minimizing the MINE loss function $\mathcal{L}_{\mathrm{mine}}$. 

SCFED follows GIB through minimizing $\mathcal{L}_{\mathrm{CL}}$ and $\mathcal{L}_{\mathrm{MI}}$, thus enhancing the detection performance from two aspects:
\begin{itemize}
    \item[1. ]Reducing the redundant mutual information shared by $Z_{G}$ and $Z_{S}$, which is irrelevant to the label prediction. 
    \item[2. ]Forcing $Z_{S}$ to be independent to $Z_{G}$ thus keeping the invariance of $Z_{S}$, thus allowing the detector to learn information that is most relevant to the attack and independent of the specific implementation of the smart contract.  
\end{itemize}

We provide detailed theoretical discussions of the above claims in Sec.~\ref{sec:theoretical}.

Overall, SCFED is trained by a bi-level optimization:
\begin{align}\label{eq:loss_total}
    \mathcal{L} & = \mathcal{L}_{\mathrm{CL}}+\alpha\cdot\mathcal{L}_{\mathrm{CF}}+\beta\cdot\mathcal{L}_{\mathrm{MI}}+\gamma\cdot\mathcal{L}_{\mathrm{SP}}, \\
    & \mathrm{s.t.}\;t^{*}=\arg\min_{t} \;\mathcal{L}_{\mathrm{mine}}(t(h_{\mathrm{S}},h_{\mathrm{G}})),\\
    &\mathcal{D}^{*}=\arg\min_{\mathcal{D}}\,\mathcal{L},\; \mathcal{D}=\{f,g,\Phi\},
\end{align}

where $\mathcal{D}$ is the detector contains the graph encoder $f$, MLP $g$ and fully connected layer $\Phi$. $\mathcal{L}$ and $\mathcal{L}_{\mathrm{mine}}$ are optimization objectives for the outer loop and the inner loop, respectively, $\alpha, \beta, \gamma$ are hyperparameters.  $\mathcal{L}_{\mathrm{CF}}$ and $\mathcal{L}_{\mathrm{SP}}$ are responsible for encouraging SCFED to generate concise counterfactual explanations, while $\mathcal{L}_{\mathrm{CL}}$ and $\mathcal{L}_{\mathrm{MI}}$ ensure that the model effectively combines graph embedding and subgraph embedding to achieve accurate detection.

By minimizing $\mathcal{L}$, SCFED follows GIB theory to utilize graph-level information and representative subgraph-level information to empower the optimal detector $\mathcal{D}^{*}$. Theoretical guarantees are formalized in Sec.~\ref {sec:theoretical}.

\subsubsection{Theoretical Analysis}
\label{sec:theoretical}
We demonstrate that SCFED follows Graph Information Bottleneck (GIB) \cite{GIB} to leverage both graph-level and subgraph-level information for detection and simultaneously reduce redundant information.

GIB aims to learn node or graph representations that retain the most relevant information for a target task while discarding irrelevant details from the input graph, thereby improving generalization in graph-based learning tasks. To begin with, we indicate that SCFED essentially learns a powerful detector $\mathcal{D}^{*}$ through the following ideal training objective:
\begin{equation}\label{eq:ideal_obj}
    \mathcal{D}^{*}=\arg\max_{\mathcal{D}}\;I(Z_{\mathrm{G}};Y)+I(Z_{\mathrm{S}};Y|Z_{\mathrm{G}})- I(Z_{\mathrm{S}};Z_{\mathrm{G}}|Y).
\end{equation}

The meaning of each term of Eq.~\ref{eq:ideal_obj} is as follows:
\begin{itemize}
    \item[1. ]$I(Z_{\mathrm{G}};Y)$ measure the mutual information shared by $Z_{\mathrm{G}}$ and $Y$, maximizing it to ensure that $Z_{\mathrm{G}}$ captures the most discriminative information to predict $Y$.
    \item[2. ]$I(Z_{\mathrm{S}};Y|Z_{\mathrm{G}})$ represents conditional mutual information between $Z_{\mathrm{S}}$ and $Y$ , given $Z_{\mathrm{G}}$. It quantities information shared by $Z_{\mathrm{S}}$ and $Y$ that is irrelevant to $Z_{\mathrm{G}}$. Maximizing it encourages $Z_{\mathrm{S}}$ to extract complementary predictive information that independent to the specific contract. 
    \item[3. ]$I(Z_{\mathrm{S}};Z_{\mathrm{G}}|Y)$ is the conditional mutual information between $Z_{\mathrm{S}}$ and $Z_{\mathrm{G}}$, given $Y$. It measures label-irrelevant information shared by $Z_{\mathrm{S}}$ and $Z_{\mathrm{G}}$. Minimizing it forces $\mathcal{D}^{*}$ to reduce redundant information, and thus focus on the most representative information for detection.
\end{itemize}

Eq.~\ref{eq:ideal_obj} indicates that the optimal $\mathcal{D}^{*}$ not only maximizes the information shared by $Z_{\mathrm{G}}$ and $Y$ (as GNN-based detectors typically do), but also extracts additional information from the factual subgraph to enhance the detection. Besides, $\mathcal{D}^{*}$ prevents the potential performance degradation caused by redundant information between $Z_{\mathrm{G}}$ and $Z_{\mathrm{S}}$.

Since direct computation of conditional mutual information is intractable, we instead achieving the ideal training objective shown in Eq.~\ref{eq:ideal_obj} by maximizing its lower bound:
\begin{equation}\label{eq:lower_bound}
	\mathcal{D}^{*} = \arg\max_{\mathcal{D}}\; I(\psi(Z_{\mathrm{S}},Z_{\mathrm{G}});Y)-I(Z_{\mathrm{S}};Z_{\mathrm{G}}),
\end{equation}

We provide Theorem.~\ref{theorem:lower_bound} to support this claim:
\begin{theorem}\label{theorem:lower_bound}
    Given the linear mapping function $\psi$, maximizing Eq.~\eqref{eq:lower_bound} is equal to maximize Eq.~\eqref{eq:ideal_obj}.
\end{theorem}

\begin{proof}
    Since $\psi$ is a linear mapping, according to the principles of mutual information, we have:
    \begin{align}
    	I(\psi(Z_{\mathrm{S}},Z_{\mathrm{G}});Y) \leq I((Z_{\mathrm{S}},Z_{\mathrm{G}});Y).
    \end{align}
    
    $(Z_{\mathrm{S}},Z_{\mathrm{G}})$ is the joint distribution of $Z_{\mathrm{S}}$ and $Z_{\mathrm{G}}$. The following inequality holds:
    \begin{equation}\label{eq:ineq1}
        I(\psi(Z_{\mathrm{S}},Z_{\mathrm{G}});Y)-I(Z_{\mathrm{S}};Z_{\mathrm{G}})\leq I((Z_{\mathrm{S}},Z_{\mathrm{G}});Y)-I(Z_{\mathrm{S}};Z_{\mathrm{G}}).
    \end{equation}

    According to the principles of mutual information, we rewrite the right side of the above inequality:
    \begin{align}\label{eq:ineq2}
        \notag&\;I((Z_{\mathrm{S}},Z_{\mathrm{G}});Y)-I(Z_{\mathrm{S}};Z_{\mathrm{G}})\\
        \notag&=I(Z_{\mathrm{G}};Y)+I(Z_{\mathrm{S}};Y|Z_{\mathrm{G}})-I(Z_{\mathrm{S}};Z_{\mathrm{G}})\\
        &\leq I(Z_{\mathrm{G}};Y)+I(Z_{\mathrm{S}};Y|Z_{\mathrm{G}})-I(Z_{\mathrm{S}};Z_{\mathrm{G}}|Y).
    \end{align}
    
    Based on Eq.~\eqref{eq:ineq1} and Eq.~\eqref{eq:ineq2}, we have:
    \begin{align}
        \notag I(\psi(Z_{\mathrm{S}},Z_{\mathrm{G}});Y)-I(Z_{\mathrm{S}};Z_{\mathrm{G}})\leq I(Z_{\mathrm{G}};Y)&+I(Z_{\mathrm{S}};Y|Z_{\mathrm{G}})\\
        &-I(Z_{\mathrm{S}};Z_{\mathrm{G}}|Y),
    \end{align}

    which indicates Eq.~\eqref{eq:lower_bound} is the lower bound of Eq.~\eqref{eq:ideal_obj}. Therefore, maximizing Eq.~\eqref{eq:lower_bound} can maximize Eq.~\eqref{eq:ideal_obj} and achieve the ideal training objective.
\end{proof}

Next, we demonstrate that SCFED maximizes Eq.~\eqref{eq:lower_bound} by minimizing $\mathcal{L}_{\mathrm{CL}}$ and $\mathcal{L}_{\mathrm{MI}}$:
\begin{theorem}\label{theorem:loss}
    Minimizing $\mathcal{L}_{\mathrm{CL}}$ and $\mathcal{L}_{\mathrm{MI}}$ maximizes Eq.~\eqref{eq:lower_bound}.
\end{theorem}

\begin{proof}
    Given $Z_{\mathrm{M}}=\psi(Z_{\mathrm{S}},Z_{\mathrm{G}})$, to implement the first term $I(Z_{\mathrm{M}};Y)$ of Eq.~\eqref{eq:lower_bound}, we leverage the cross-entropy loss to implement $\mathcal{L}_{\mathrm{CL}}$ since the cross-entropy loss is equal to the conditional information entropy $H(Y|Z_{\mathrm{M}})$:
    
    \begin{equation}
        \mathcal{L}_{\mathrm{CL}}(Z_{\mathrm{M}},Y) = H(Y|Z_{\mathrm{M}}) = \mathbb{E} \left[ -\log \mathrm{P}(Y|Z_{\mathrm{M}}) \right].
    \end{equation}
    
    According to the principle of mutual information, we derive:
    \begin{align}
        I(h_{\mathrm{M}}; Y) &= H(Y)-H(Y|Z_{\mathrm{M}}),\\
        &=  H(Y)-\mathcal{L}_{\mathrm{CL}}(Z_{\mathrm{M}},Y).
    \end{align}
    
    As $H(Y)$ constant given the label distribution, minimizing $\mathcal{L}_{\mathrm{CL}}$ is equal to maximizing $I(Z_{\mathrm{M}}; Y)$:
    \begin{equation}
        \min\; \mathcal{L}_{\mathrm{CL}}(Z_{\mathrm{M}},Y) \implies \max\; I(Z_{\mathrm{M}}; Y).
    \end{equation}
    
    As for the second term of Eq.~\eqref{eq:lower_bound}, we utilize MINE~\cite{MINE} $t: \mathbb{R}^{d}\times \mathbb{R}^{d}\rightarrow \mathbb{R}$ to estimate the value of $I(Z_{\mathrm{S}};Z_{\mathrm{G}})$ since it is hard to compute the mutual information directly. A trained estimator $t^{*}$ outputs the estimated value $t^{*}(Z_{\mathrm{S}},Z_{\mathrm{G}})$. Therefore, the minimization of $I(Z_{\mathrm{S}};Z_{\mathrm{G}})$ is equal to minimize the estimated value:
    \begin{align}
        \min\;& \mathcal{L}_{\mathrm{MI}}  \implies \min\; I (Z_{\mathrm{S}};Z_{\mathrm{G}}),\quad\mathcal{L}_{\mathrm{MI}} = t^{*}(Z_{\mathrm{S}},Z_{\mathrm{G}}), \\
        & \mathrm{s.t.}\; t^{*}=\arg\min_{t} \;\mathcal{L}_{\mathrm{mine}}(t(Z_{\mathrm{S}},Z_{\mathrm{G}})).
    \end{align}

    Overall, we demonstrate:
    \begin{align}
        \arg\min_{\mathcal{D}}\;\mathcal{L}_{\mathrm{CL}}+\mathcal{L}_{\mathrm{MI}}\implies \arg\max_{\mathcal{D}}\;I(Z_{\mathrm{M}};Y)-I(Z_{\mathrm{S}};Z_{\mathrm{G}}).
    \end{align}
\end{proof}

\section{Experiments}

In this section, we evaluate SEASONED through extensive experiments addressing three research questions:

\noindent \textbf{RQ1 Detection \& Explainability:} How effectively does SEASONED detect Adversarial Exploiter Contracts (AECs) while generating counterfactual explanations? (\textbf{Obj.1})

\noindent \textbf{RQ2 Multidimensional Performance Enhancement:} To what extent does SEASONED improve robustness, generalizability, and data efficiency? (\textbf{Obj.2})

\noindent \textbf{RQ3 Component Efficacy:} What is the individual contribution of SRG construction and SCFED's explanation-generation mechanism to overall performance?

\subsection{Experimental Settings}
\paragraph{Implementation}
Our experiment environment is comprised of a server with a 12-core Intel(R)-Core-12700k CPU @3.61 GHz, an NVIDIA GeForce RTX 3090 Ti GPU, 128GB of RAM, and the Ubuntu 24.04 LTS operating system.

For the construction of SRGs, we employ the web3.py library to extract contract bytecode and Vandal~\cite{brent_vandal_2018} to generate the intermediate representation. Subsequently, we construct the SRG by capturing semantic relationships within the IR.
For SCFED's implementation, we use a 2-layer RGCN\cite{schlichtkrull2018modeling} with 8-dimensional hidden states as the graph encoder $f$, since RGCN is well-suited for capturing semantic relations information present in the SRG. Additionally, we employ a 2-layer MLP with 8-dimensional hidden states as $g$. The functions $\Phi$ and $\psi$ are implemented as a fully-connected layer and a mean function, respectively. SCFED is trained using the Adam optimizer with a learning rate of 0.001. The training process employs a bi-level optimization scheme with 50 epochs for the outer loop and 100 epochs for the inner loop. All experiments adopt the hyperparameters $\{\alpha=0.5, \beta=0.5, \gamma=0.1, \rho=0.5\}$. We conduct 10-fold cross-validation and report the average evaluation results.

\paragraph{Dataset}
To compose a comprehensive dataset, we collected AECs used in real-world attack incidents on Ethereum from 5 public databases: DefiLlama~\cite{defillama}, rekt news~\cite{rekt}, SlowMist~\cite{slowmist}, Neptune Mutual~\cite{neptune}, and ChainSec~\cite{chainsec}.
We gathered a total of 359 AECs involved in attacks occurring between June 2016 and January 2025.
It is important to note that our dataset focuses exclusively on contracts with active attack intent and primary attack logic, excluding contracts associated with general fraud or phishing, such as Ponzi schemes and phishing attacks.
For benign contracts, we selected contracts 
from Forta's benign contracts dataset to ensure diversity in sample types.
The final dataset comprises 359 AECs and 1,196 benign contracts.




\paragraph{Metrics}
We evaluate detection performance using Precision (\textbf{P}), Recall (\textbf{R}), and F1-score (\textbf{F1}), which respectively measure the accuracy, completeness, and overall balance of positive samples identification. For explanation evaluation, following \cite{cf2} and \cite{nseg}, we use probability of necessity (\textbf{PN}), probability of sufficiency (\textbf{PS}), and \textbf{Avg.Size}, the average size of factual subgraphs.


\paragraph{Baselines}
To evaluate the detection and explanation effectiveness of the proposed method, we compare it with three categories of baselines:
\begin{itemize}
    \item \textbf{State-of-the-Art (SOTA) AECD Methods}: We compare with three leading approaches: \textit{Forta}~\cite{forta2023}, \textit{LookAhead}~\cite{ren2024lookahead}, and \textit{Skyeye}~\cite{wang2024skyeye} to benchmark our method against the current SOTA methods in adversarial exploiter contracts detection.
    \item \textbf{GNN-based Detectors}: We apply \textit{GCN}~\cite{GCN}, \textit{GAT}~\cite{GAT}, \textit{GIN}~\cite{GIN} and \textit{RGCN}~\cite{schlichtkrull2018modeling} models to the constructed SRGs for detection.
    \item \textbf{Post-hoc Explainers}: We compare our methods with several representative post-hoc GNN explainers, including \textit{GNNExplainer}~\cite{gnnexplainer}, \textit{PGExplainer}~\cite{pgexplainer}, the counterfactual explainer \textit{CF$^2$}~\cite{cf2}, and the \textit{Necessary and Sufficient Explanation for Graph (NSEG)}~\cite{nseg}.
\end{itemize}


    
    
    

\subsection{Detection and Explanation Results (RQ1)}
\label{sec:expl_results}
\paragraph{Detection Results}
As shown in Table.~\ref{tab:result}, F1 score of our approach is over 0.99, which significantly outperforms SOTA AECD methods, and the inferior performance of other GNN-based detectors compared to RGCN+SRG underscores the importance of learning distinct relationship types in the semantic relation graph. 

\begin{table}[!ht]
    \centering
    \small
    \begin{tabular}{ l |ccc}
        \toprule
        Methods & \textbf{Precision} & \textbf{Recall} & \textbf{F1-Score} \\
        \midrule
        Forta \cite{forta2023} & 0.8860 & 0.5940 & 0.7112 \\
        LookAhead \cite{ren2024lookahead} & 0.9273 & 0.9623 & 0.9444 \\
        Skyeye \cite{wang2024skyeye} & 0.9454 & 0.9286 & 0.9369 \\
        \midrule
        GCN+SRG & 0.9012 & 0.8794 & 0.8823 \\
        GAT+SRG & \underline{0.9551} & 0.8919 & 0.9108 \\
        GIN+SRG & 0.9524 & 0.9581 & 0.9570 \\
        RGCN+SRG & \textbf{1.0} & 0.9918 & 0.9901 \\
        \midrule
        SEASONED & \textbf{1.0} & \textbf{0.9976} & \textbf{0.9988} \\ 
        \bottomrule
    \end{tabular}
    \caption{Detection comparison.}
    \label{tab:result}
\end{table}

The performance gap between SEASONED and SOTA methods may stem from the reliance of existing AECD methods on syntactical features, which may fail to fully capture latent dependencies that reveal semantic features within contract bytecode.
In contrast, our proposed SRG effectively captures essential semantic relationships, thereby enabling more accurate detection of adversarial exploiter contracts.

\paragraph{Explanation Results}
As presented in Table~\ref{tab:explain}, SEASONED demonstrates superior performance compared to post-hoc explainers. GNNExplainer and PGExplainer yield the lowest PN values because they solely pursue label consistency between explanation subgraphs and graphs, while disregarding residual graph labels. Although CF$^2$ and NSEG achieve relatively high PS values, their PN performance remains limited. SEASONED produces comparable but more balanced results, delivering satisfactory explanation performance. Notably, while SEASONED does not achieve the smallest average explanation size (Avg.Size), we argue that its dual optimization objective effectively balances concise explanation with representative information learning, ensuring both interpretability (Obj.1) and detection enhancement (Obj.2).

\begin{table}[!ht]
    \centering
    \small
    \begin{tabular}{ l | c c c }
        \toprule
        Explainers & \textbf{PS} & \textbf{PN} & \textbf{Avg.Size} \\
        \midrule
        GNNExlainer & 0.85 & 0.54 & 714.3 \\
        PGExlainer & 0.89 & 0.77 & 662.0 \\
        CF$^2$ & 0.96 & \textbf{1.0} & 571.8 \\
        NSEG & \underline{0.98} & 0.97 & \underline{492.2} \\
        \midrule
        SEASONED w/o $\mathcal{L}_{\mathrm{CF}}$  & 0.84 & 0.15 & \textbf{421.7} \\
        SEASONED w/o $\mathcal{L}_{\mathrm{SP}}$  & \textbf{0.99} & 0.74 & 1,150.6 \\
        SEASONED w/o $\mathcal{L}_{\mathrm{MI}}$  & 0.94 & \underline{0.99} & 688.4 \\
        \midrule
        SEASONED & 0.96 & \textbf{1.0} & 646.9 \\
        \bottomrule
    \end{tabular}
    \caption{Explanation comparison.}
    \label{tab:explain}
\end{table}

\begin{table}[!ht]
\centering
\resizebox{\linewidth}{!}{
    \begin{tabular}{c|lc|lc}
    \toprule
    \textbf{Feature} & \multicolumn{2}{c|}{\textbf{Benign Contracts}} & \multicolumn{2}{c}{\textbf{AECs}} \\ \midrule
    \multirow{5}{*}{\makecell[c]{Top 5 Opcodes\\(Opcode \& Ratio)}} & CONST & 0.27 & CONST & 0.32 \\
    & ADD & 0.10 & \textbf{JUMPDEST} & 0.09 \\
    & \textbf{MLOAD} & 0.09 & ADD & 0.08 \\
    & \textbf{MSTORE} & 0.08 & \textbf{JUMP} & 0.05 \\
    & JUMPI & 0.07 & JUMPI & 0.05 \\ \midrule
    \multirow{3}{*}{Relation \& Ratio} & Data & 0.68 & Data & 0.60 \\
    & Control & \textbf{0.19} & Control & \textbf{0.34} \\
    & Effect & 0.13 & Effect & 0.06 \\ \midrule
    Avg. Path Length & \multicolumn{2}{c|}{\textbf{3.61}} & \multicolumn{2}{c}{\textbf{10.52}} \\\bottomrule
    \end{tabular}}
\caption{Explanation Analysis.}
\label{tab:subgraph_stats}
\end{table}

Furthermore, we analyze the explanation results of SEASONED and present the statistical comparison in Table~\ref{tab:subgraph_stats}, which reveals three key observations:

\begin{itemize}
\item \textbf{Opcode Distribution Divergence}: Comparing top-5 most frequent opcodes in factual subgraphs, AECs exhibit higher proportions of \texttt{JUMP}-related opcodes (\texttt{JUMP}, \texttt{JUMPI}, \texttt{JUMPDEST}), whereas benign contracts are dominated by data operations (\texttt{MLOAD}, \texttt{MSTORE}).

\item \textbf{Relation Type Bias}: AECs place greater emphasis on control relations (0.34 vs 0.19 for benign contracts) while exhibiting fewer effect relations (0.06 vs 0.13).

\item \textbf{Structural Complexity Gap}: The subgraphs of AECs demonstrate an average path length that is 2.9 times longer than that of benign contracts (10.52 vs 3.61), indicating that AECs tend to have longer execution paths.
\end{itemize}

In summary, AECs show a stronger bias toward long execution paths and control-flow manipulation, whereas benign contracts focus more on memory operations.

\begin{table*}[ht]
\centering
\small
\begin{tabular}{l|ccc|ccc|ccc|ccc|ccc}
\toprule
    \multirow{2}{*}{\shortstack{Attacks\\$K$-$M$}} & \multicolumn{3}{c|}{GCN+SRG} & \multicolumn{3}{c|}{GAT+SRG} & \multicolumn{3}{c|}{GIN+SRG} & \multicolumn{3}{c|}{RGCN+SRG} & \multicolumn{3}{c}{SEASONED} \\\cline{2-16}
    \addlinespace[0.3ex]
    & P & R & F1 & P & R & F1 & P & R & F1 & P & R & F1 & P & R & F1 \\ \hline
    \textbf{GIA} & & & & & & & & & & & & & & & \\
    \;\,50\%-5 & 0.90 & 0.86 & 0.87 & 0.94 & 0.95 & 0.95 & 0.95 & 0.89 & 0.90 & \textbf{1.0} & \underline{0.98} & \underline{0.99} & \textbf{1.0} & \textbf{0.99} &\textbf{0.99} \\
    \;\,50\%-10 & 0.88 & 0.83 & 0.85 & 0.90 & 0.92 & 0.91 & 0.92 & 0.87 & 0.88 & \textbf{1.0} & \underline{0.91} & \underline{0.95} & \textbf{1.0} & \textbf{0.99} & \textbf{0.99} \\
    \;\,100\%-5 & 0.88 & 0.85 & 0.86 & 0.91 & 0.90 & 0.90 & 0.90 & 0.86 & 0.87 & \underline{0.99} & \underline{0.94} & \underline{0.96} & \textbf{1.0} & \textbf{0.99} & \textbf{0.99} \\
    \;\,100\%-10 & 0.82 & 0.78 & 0.79 & 0.87 & 0.87 & 0.87 & 0.89 & 0.83 & 0.84 & \textbf{1.0} & \underline{0.87} & \underline{0.93} & \textbf{1.0} & \textbf{0.99} & \textbf{0.99} \\
    \textbf{LFA} & & & & & & & & & & & & & & & \\
    \;\,30\% & 0.89 & 0.81 & 0.85 & 0.95 & 0.89 & 0.92 & 0.92 & 0.85 & 0.88 & \textbf{1.0} & \underline{0.93} & \underline{0.96} & \textbf{1.0} & \textbf{0.94} & \textbf{0.97} \\
    \;\,50\% & 0.86 & 0.50 & 0.63 & 0.93 & 0.53 & 0.68 & 0.89 & 0.48 & 0.62 & \textbf{1.0} & \underline{0.67} & \underline{0.80} & \textbf{1.0} & \textbf{0.76} & \textbf{0.86} \\
    \;\,70\% & 0.86 & 0.14 & 0.24 & 0.92 & 0.18 & 0.30 & 0.84 & 0.12 & 0.21 & \textbf{1.0} & \underline{0.20} & \underline{0.33} & \textbf{1.0} & \textbf{0.41} & \textbf{0.58} \\
    \textbf{PRBCD} & & & & & & & & & & & & & & & \\
    \;\,1\% & 0.74 & 0.43 & 0.54 & 0.86 & 0.56 & 0.68 & 0.81 & 0.49 & 0.61 & \textbf{1.0} & \underline{0.57} & \underline{0.73} & \textbf{1.0} & \textbf{0.89} & \textbf{0.94} \\
    \;\,3\% & 0.60 & 0.24 & 0.34 & 0.71 & 0.37 & 0.49 & 0.74 & \underline{0.40} & \underline{0.52} & \textbf{1.0} & 0.31 & 0.48 & \textbf{1.0} & \textbf{0.54} & \textbf{0.70} \\
    \;\,5\% & 0.55 & 0.17 & 0.26 & 0.62 & 0.25 & 0.36 & 0.59 & 0.24 & 0.34 & \textbf{1.0} & \underline{0.26} & \underline{0.41} & \textbf{1.0} & \textbf{0.51} & \textbf{0.68} \\
\bottomrule
\end{tabular}
\caption{Robustness comparison.}
\label{tab:robust}
\end{table*}

\begin{table*}[ht]
\centering
\small
    \begin{tabular}{l|ccc|ccc|ccc|ccc|ccc}
    \toprule
    \multirow{2}{*}{\shortstack{Training\\Set}} & \multicolumn{3}{c|}{GCN+SRG} & \multicolumn{3}{c|}{GAT+SRG} & \multicolumn{3}{c|}{GIN+SRG} & \multicolumn{3}{c|}{RGCN+SRG} & \multicolumn{3}{c}{SEASONED} \\\cline{2-16}
    \addlinespace[0.3ex]
    & P & R & F1 & P & R & F1 & P & R & F1 & P & R & F1 & P & R & F1 \\ \hline
    \addlinespace[0.3ex]
    Rand-10\% & 0.86 & 0.84 & 0.85 & 0.94 & 0.89 & 0.91 & 0.92 & 0.90 & 0.91 & \underline{0.96} & \textbf{1.0} & \underline{0.98} & \textbf{0.98} & \underline{0.99} & \textbf{0.99} \\
    Rand-20\% & 0.90 & 0.88 & 0.89 & \underline{0.95} & 0.89 & 0.92 & \underline{0.95} & \underline{0.95} & \underline{0.95} & \textbf{1.0} & \textbf{0.99} & \textbf{0.99} & \textbf{1.0} & \textbf{0.99} & \textbf{0.99} \\ 
    \addlinespace[0.3ex]
    \hline
    \addlinespace[0.3ex]
    Old-10\% & 0.66 & 0.52 & 0.58 & 0.51 & 0.76 & 0.61 & \underline{0.83} & 0.65 & \underline{0.73} & 0.48 & \textbf{1.0} & 0.64 & \textbf{0.88} & \underline{0.99} & \textbf{0.93} \\
    Old-20\% & 0.89 & 0.87 & 0.88 & 0.94 & 0.87 & 0.90 & 0.93 & 0.94 & 0.93 & \underline{0.95} & \textbf{1.0} & \underline{0.97} & \textbf{0.97} & \underline{0.99} & \textbf{0.98} \\
    \bottomrule
    \end{tabular}
\caption{Generalizability comparison.}
\label{tab:general}
\end{table*}

\begin{table*}[ht]
\centering
\small
    \begin{tabular}{c|ccc|ccc|ccc|ccc|ccc}
    \toprule
    \multirow{2}{*}{Samples} & \multicolumn{3}{c|}{GCN+SRG} & \multicolumn{3}{c|}{GAT+SRG} & \multicolumn{3}{c|}{GIN+SRG} & \multicolumn{3}{c|}{RGCN+SRG} & \multicolumn{3}{c}{SEASONED} \\\cline{2-16}
    \addlinespace[0.3ex]
    & P & R & F1 & P & R & F1 & P & R & F1 & P & R & F1 & P & R & F1 \\ \hline
    \addlinespace[0.3ex]
    20 & 0.78 & 0.21 & 0.33 & 0.82 & 0.27 & 0.41 & 0.80 & 0.28 & 0.41 & \underline{0.92} & \underline{0.34} & \underline{0.50} & \textbf{0.94} & \textbf{0.43} & \textbf{0.59} \\
    30 & 0.80 & 0.17 & 0.28 & 0.82 & 0.27 & 0.41 & 0.84 & 0.30 & 0.44 & \underline{0.94} & \underline{0.41} & \underline{0.57} & \textbf{0.98} & \textbf{0.52} & \textbf{0.68} \\
    40 & 0.82 & 0.25 & 0.38 & 0.84 & 0.29 & 0.43 & 0.84 & 0.32 & 0.46 & \underline{0.95} & \underline{0.45} & \underline{0.61} & \textbf{0.99} & \textbf{0.91} & \textbf{0.95} \\
    50 & 0.82 & 0.27 & 0.41 & 0.88 & 0.31 & 0.46 & 0.89 & 0.33 & 0.48 & \underline{0.95} & \underline{0.48} & \underline{0.64} & \textbf{1.0} & \textbf{0.96} & \textbf{0.98} \\
    \bottomrule
    \end{tabular}
\caption{Data-efficient learning comparison.}
\label{tab:dataeff}
\end{table*}

\subsection{Robustness Evaluation (RQ2)}
\paragraph{Attack Settings} Although attackers cannot directly access the SRG, they can indirectly modify the SRG by altering the source code. Since we do not know the exact content of the source code, the resulting changes to the SRG after such modifications are unknown and uncontrollable. Therefore, we employ three graph-domain adversarial attack methods to modify the SRG, in order to evaluate the robustness of our model when facing evading techniques that might be adopted by attackers.

\textbf{Graph Injection Attack (GIA)}: This simulates real-world scenarios where attackers insert benign or irrelevant code segments into AECs to dilute the proportion of attacker code segments, effectively injecting nodes and edges into the contract’s corresponding SRG. Since there are no existing GIAs against AECs detection, we randomly insert $K\%$ of the original SRG nodes ($K=\{50,100\}$), with each inserted node randomly connected to $M$ existing nodes ($M=\{3,5,10\}$).

\textbf{Label Flipping Attack (LFA)}: This represents cases where manual contract labeling errors occur, leading to flipped labels. We follow the settings of \cite{lfa1,lfa2}, randomly flipping the labels of $K\%$ samples in the training set ($K=\{30,50,70\}$).

\textbf{Structural Attack}: Considering our large-scale dataset (containing thousands of SRGs with millions of nodes in total), we employ the efficient evasion-based structural attack methods: Projected Randomized Block Coordinate Descent (\textbf{PRBCD}), proposed in \cite{grbcd}, to generate modified graphs. We strictly adhere to the configurations in \cite{grbcd}, modifying each SRG by inserting or removing $K\%$ edges ($K=\{1,3,5\}$). 

\paragraph{Robustness Comparison} Table~\ref{tab:robust} showcases SEASONED’s superior robustness across all attack scenarios. Firstly, SEASONED exhibits complete immunity to random GIA (Recall = 0.99), whereas RGCN experiences a significant performance drop (Recall = 0.87 under 100\% node and 10-edge insertion). This observation confirms that while random code segment injections can effectively evade RGCN detection, SEASONED maintains perfect accuracy regardless of the intensity of the attack. Secondly, SEASONED exhibits significantly stronger robustness against LFA than GNN-based detectors, achieving 0.41 recall at 70\% label noise, which doubles the second-best performance (RGCN's). This robustness advantage is similarly evident in PRBCD attacks, where SEASONED consistently outperforms all baseline methods.

Notably, RGCN and SEASONED share the same graph encoder. Therefore, the divergence in robustness between them demonstrates that SEASONED enhances robustness by extracting representative information from explanations.

\subsection{Generalizability Evaluation (RQ2)}
Table.~\ref{tab:general} presents model performance comparisons under two training set sampling strategies to evaluate generalizability. The Rand-$K$\% strategy employs random sampling of $K$\% contracts from the dataset, while the Old-$K$\% approach selects the oldest $K$\% contracts based on deployment timestamps ($K=\{10,20\}$). The remaining contracts constitute the validation and test sets. This experimental design specifically tests whether detectors trained on historical contracts maintain effectiveness when applied to newly deployed contracts.

Table.~\ref{tab:general} shows that GNN-based detectors perform significantly worse on Old-10\% than on Rand-10\%, despite both containing the same number of samples. This performance gap reveals their poor generalizability to new contracts. In contrast, SEASONED achieves much better results (F1=0.93 on Old-10\%), substantially outperforming the second-best detector (F1=0.73). These results demonstrate SEASONED's superior ability to generalize to newly deployed contracts, making it more suitable for real-world Ethereum applications. 

We attribute SEASONED's generalizability to its unique mechanism of extracting representative information from explanations that remain invariant to Ethereum's temporal evolution. This approach enables consistent detection performance across contracts in different periods. 

\subsection{Data-efficient Learning Evaluation (RQ2)}
Considering the scarcity of well-labeled contracts in real-world scenarios, the data-efficient learning capability becomes crucial for detectors. We evaluate this by constructing a training set through random sampling of $K$ contracts ($K=\{20,30,40,50\}$), with results presented in Table~\ref{tab:dataeff}. As shown in Table~\ref{tab:dataeff}, SEASONED achieves comparable performance to full-dataset training when using limited data, attaining F1 scores of 0.95 (40 samples) and 0.98 (50 samples), while other detectors' F1 scores are below 0.7. This proves the data-efficient learning ability of SEASONED that addresses the critical challenge of label scarcity in smart contract analysis.

\subsection{Ablation Studies (RQ3)}
\paragraph{Effectiveness of SRG}
To evaluate the effectiveness of SRG, we performed an ablation study by deconstructing this graph representation, which consists of three relation types, into different combinations.

\begin{figure}[t]
    \centering
    \includegraphics[width=0.9\linewidth]{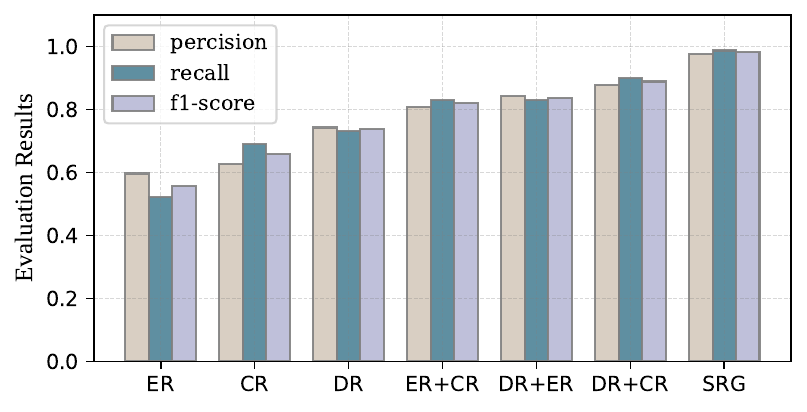}
    \caption{Ablation study to evaluate the effectiveness of SRG.}
    \label{fig:ablation}
\end{figure}

Specifically, we examined seven combinations: data relation only (DR), effect relation only (ER), control relation only (CR), data and effect relation (DR+ER), data and control relation (DR+CR), effect and control relation (ER+CR), and the complete SRG encompassing all three relations. Evaluation results are presented in \ref{fig:ablation}. Notably, the evaluation metrics for individual relationship types drop significantly compared to the original SRG, even for the Data Flow Relation, which constitutes the largest proportion of edges in the SRG. Graphs retaining two types of edges show improved metrics over single-edge graphs but remain insufficient for fully detecting malicious behavior. For the ER+DR configuration, the absence of control flow edges likely results in missing contextual information, such as function call order. Similarly, the ER+CR graph lacks the dominant data flow edges, leading to insufficient capture of semantic information like data propagation. The DR+CR configuration achieves better performance by capturing both execution logic and data dependencies, yet it still falls short of the comprehensive analysis enabled by SRG. This gap may be attributed to the absence of effect flow, which captures critical state changes like asset transfers. 

\paragraph{Effectiveness of SCFED}
Table.~\ref{tab:explain} and Table.~\ref{tab:ablation} show the comparison of three variants of SCFED: SCFED w/o $\mathcal{L}_{\mathrm{CF}}$, SCFED w/o $\mathcal{L}_{\mathrm{SP}}$ and SCFED w/o $\mathcal{L}_{\mathrm{MI}}$, which remove $\mathcal{L}_{\mathrm{CF}}$, $\mathcal{L}_{\mathrm{SP}}$, and $\mathcal{L}_{\mathrm{MI}}$, respectively. 

\begin{table}[!ht]
\centering
\resizebox{\linewidth}{!}{
    \begin{tabular}{l|c|cccc}
    \toprule
    \multirow{2}{*}{Variants} & \multirow{2}{*}{Metric} & \multirow{2}{*}{Clean} & GIA & LFA & PRBCD \\ 
    & & & 100\%-10 & 70\% & 5\% \\
    \midrule \addlinespace[0.3ex]
    
    \multirow{3}{*}{SCFED w/o $\mathcal{L}_{\mathrm{CF}}$} & \textbf{P} & \textbf{1.0} & \textbf{1.0} & \textbf{1.0} & \underline{0.95} \\
    & \textbf{R} & 0.99 & 0.90 & 0.34 & 0.30 \\
    & \textbf{F1} & 0.99 & 0.95 & 0.51 & 0.46 \\
    \midrule \addlinespace[0.3ex]
    
    \multirow{3}{*}{SCFED w/o $\mathcal{L}_{\mathrm{SP}}$} & \textbf{P} & \textbf{1.0} & \textbf{1.0} & \textbf{1.0} & \textbf{1.0} \\
    & \textbf{R} & 0.99 & \underline{0.99} & \underline{0.37} & \underline{0.50} \\
    & \textbf{F1} & 0.99 & \underline{0.99} & \underline{0.54} & \underline{0.67} \\
    \midrule \addlinespace[0.3ex]
    
    \multirow{3}{*}{SCFED w/o $\mathcal{L}_{\mathrm{MI}}$} & \textbf{P} & \textbf{1.0} & \textbf{1.0} & 0.93 & 0.95 \\
    & \textbf{R} & 0.99 & 0.88 & 0.23 & 0.20 \\
    & \textbf{F1} & 0.99 & 0.94 & 0.37 & 0.33 \\
    \midrule \addlinespace[0.3ex]
    
    \multirow{3}{*}{SCFED} & \textbf{P} & \textbf{1.0} & \textbf{1.0} & \textbf{1.0} & \textbf{1.0} \\
    & \textbf{R} & \textbf{1.0} & \textbf{0.99} & \textbf{0.41} & \textbf{0.51} \\
    & \textbf{F1} & \textbf{1.0} & \textbf{0.99} & \textbf{0.58} & \textbf{0.68} \\
    \bottomrule
    \end{tabular}
}
\caption{Comparison of SCFED and its ablated variants.}
\label{tab:ablation}
\end{table}

Table~\ref{tab:ablation} shows that $\mathcal{L}_{\mathrm{MI}}$ is crucial for the robustness of SCFED. SCFED w/o $\mathcal{L}_{\mathrm{MI}}$ achieves the worst performance under various attacks. This is because $\mathcal{L}_{\mathrm{MI}}$ encourages SCFED to extract representative information and filter adversarial noise and decoy patterns, thus improving robustness.

$\mathcal{L}_{\mathrm{CF}}$ and $\mathcal{L}_{\mathrm{SP}}$ have a greater influence on the explanatory effectiveness of SCFED. The PN value of SCFED w/o $\mathcal{L}_{\mathrm{CF}}$ is only 0.15, since without $\mathcal{L}_{\mathrm{CF}}$, the prediction value of the residual graph $R$ cannot be guaranteed to be zero, indicating that the explanations provided by SCFED w/o $\mathcal{L}_{\mathrm{CF}}$ are not accurate. As for SCFED w/o $\mathcal{L}_{\mathrm{SP}}$, despite achieving relatively higher PS and PN values, its Avg.Size exceeds one thousand, suggesting that its explanations contain too much redundancy.

These results indicate that $\mathcal{L}_{\mathrm{MI}}$ ensures SCFED can extract representative information that is crucial for robustness, while $\mathcal{L}_{\mathrm{CF}}$ and $\mathcal{L}_{\mathrm{SP}}$ ensure that SCFED can provide accurate and concise explanations.

\begin{figure}[t]
    \centering
    \includegraphics[width=0.9\linewidth]{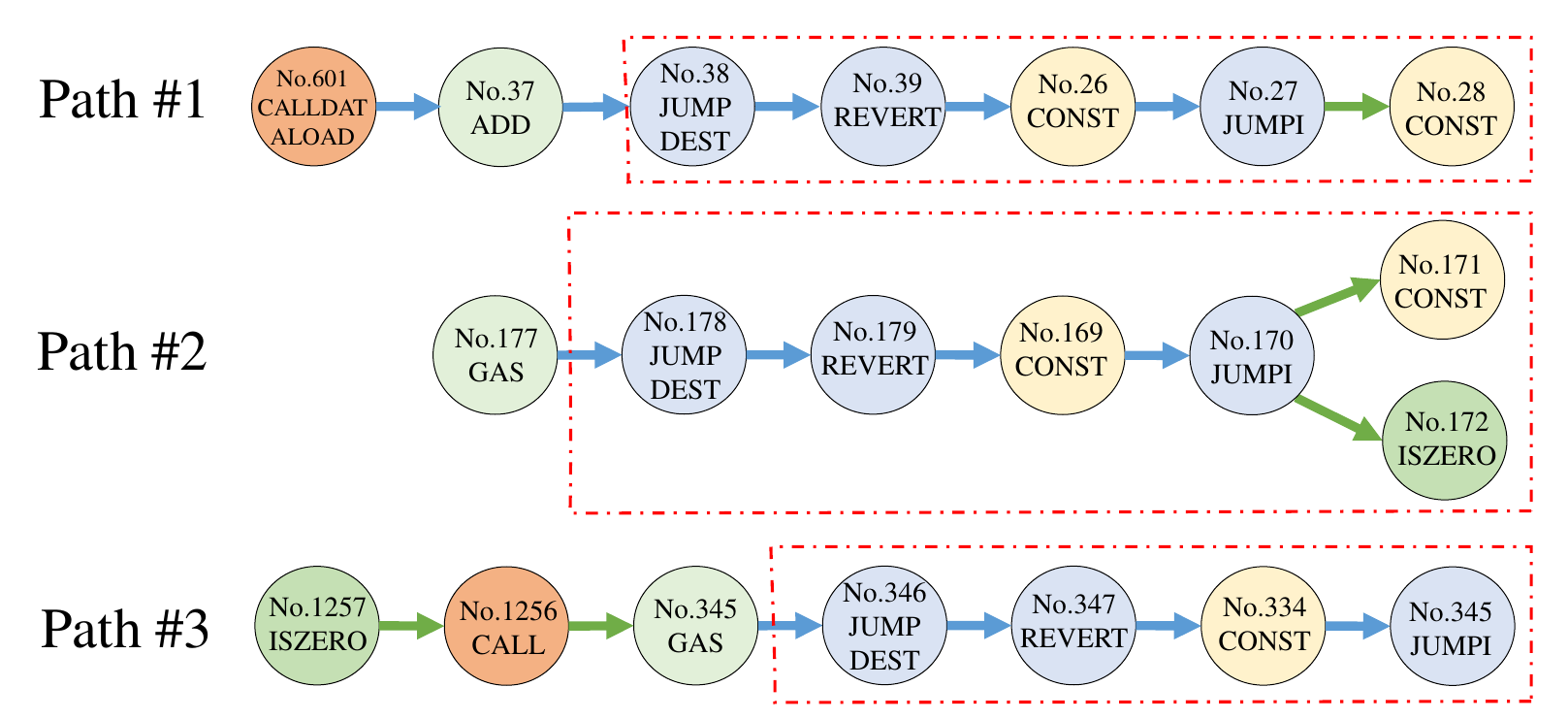}
    \caption{Examples of malicious execution paths in the factual subgraph of Harvest. Each node is marked with an index and corresponding opcode}
    \label{fig:case}
\end{figure}

\subsection{Case Study}
Since the AECs leveraged in real-world attack incidents are typically closed-source, we validated the effectiveness of our explanation framework by referencing a proof-of-concept AEC for Harvest Finance Hack~\footnote{https://github.com/abdulsamijay/Defi-Hack-Analysis-POC/blob/master/src/harvest-finance/HarvestExploit.sol}, which is developed by white-hat researchers.
Specifically, we manually label the code segments responsible for attacks and construct its corresponding SRG from the compiled bytecode. The resulting SRG contains 1,602 nodes, of which 1,016 correspond to attack code segments. SEASONED successfully identifies the entire SRG as AEC, and its generated factual subgraph, consisting of 1,272 nodes, fully covers all 1,016 nodes of attack segments. 

Fig.~\ref{fig:case} shows three example execution paths that are repeated in the factual subgraph. These paths exhibit similar subpaths (highlighted in red) and include numerous JUMP nodes and control relations (blue edges), which are not present in the counterfactual subgraph. This observation aligns with the behavioral patterns observed in the actual attack transaction \href{https://app.blocksec.com/explorer/tx/eth/0x9d093325272701d63fdafb0af2d89c7e23eaf18be1a51c580d9bce89987a2dc1}{0x9d09}, as well as with the core attack logic implemented in the proof-of-concept AEC, where repeated flashloans and swaps trigger multiple external contract and function calls.

\subsection{Visualization}
Since most AECs have thousands of nodes and are difficult to visualize, we select the two smallest AECs, 0xd09b (169 nodes) and 0x8aa5 (472 nodes), for visualization. The explanation results for these two contracts are shown in Fig.~\ref{fig:visl}(a) and Fig.~\ref{fig:visl}(d), where the factual subgraph is highlighted in orange. It can be observed that the factual subgraphs of both contracts contain long chain-like patterns, which correspond to long execution paths in the attacker contract code. This indicates that SEASONED considers long execution paths to be an important factor for detecting contracts as malicious.

We use GIA to insert nodes and edges into the SRG of both contracts, with node insertion ratios of 50\% (Fig.~\ref{fig:visl}(b) and Fig.~\ref{fig:visl}(e)) and 100\% (Fig.~\ref{fig:visl}(c) and Fig.~\ref{fig:visl}(f)). For clarity in visualization, we insert only one edge for each inserted node. If an inserted node is included in the factual subgraphs, it is marked in red; otherwise, it is marked in green. Fig.~\ref{fig:visl}(b)(c)(e) show only a few red nodes, and in Fig.~\ref{fig:visl}(f), there are 113 red nodes, accounting for only 24\% of all inserted nodes (472). This demonstrates that SEASONED can filter out interference and provide accurate explanation results.

\begin{figure}[!t]
\vspace{-0.4cm}
\centering
\begin{subfigure}{0.15\textwidth}
    \includegraphics[width=1.1\textwidth]{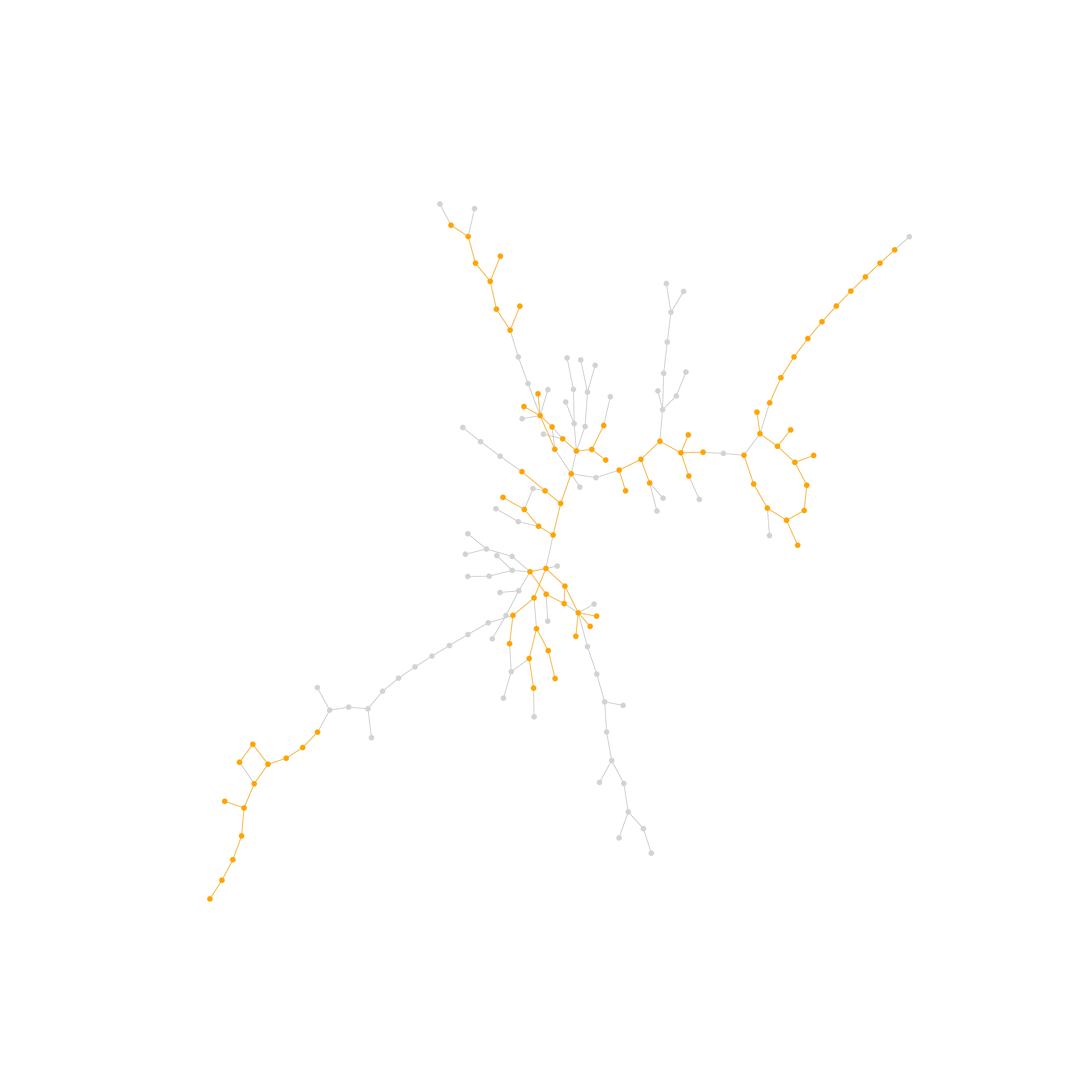}
    \caption{0xd09b Clean}
\end{subfigure}
\hfill
\begin{subfigure}{0.15\textwidth}
    \includegraphics[width=1.1\textwidth]{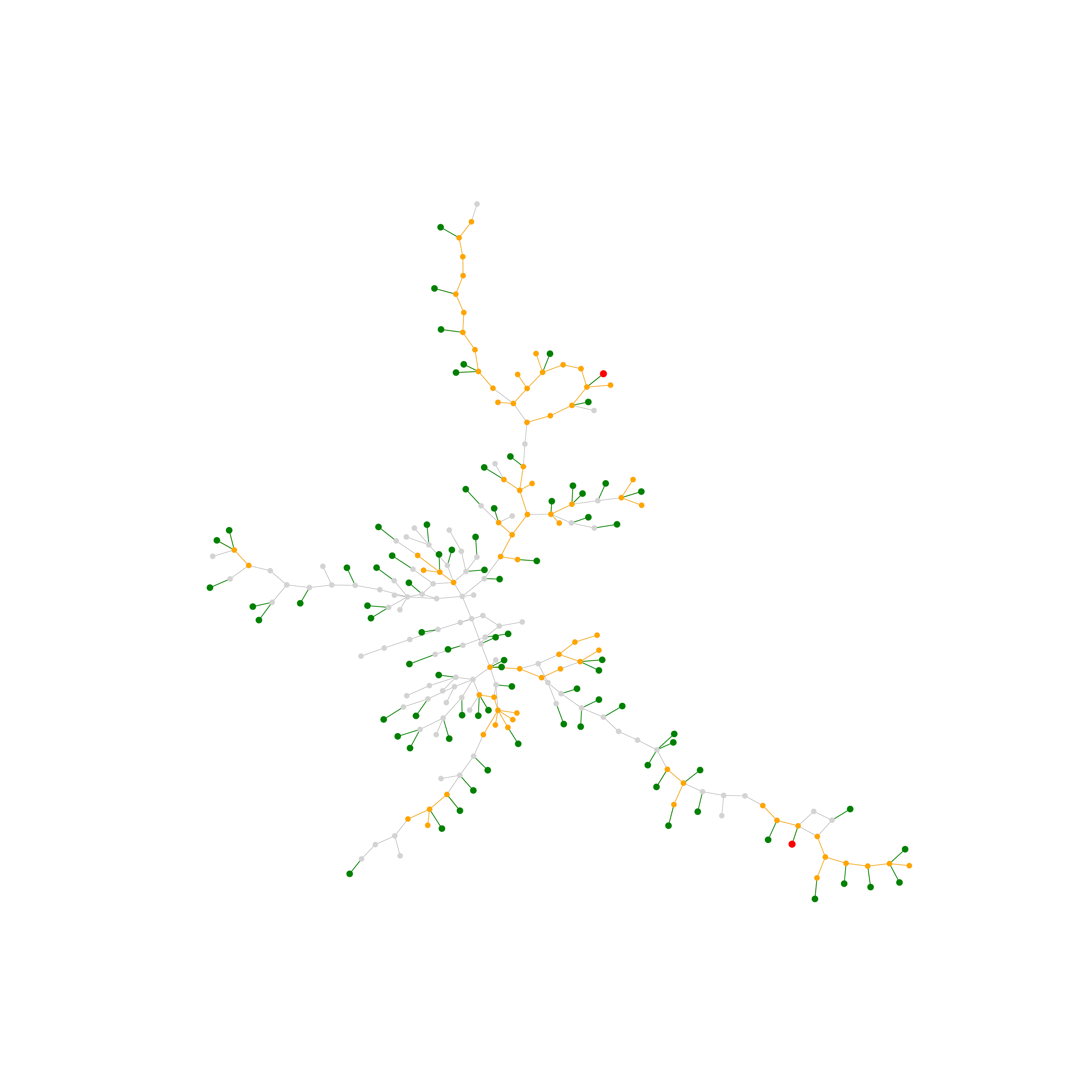}
    \caption{0xd09b 50\%-1}
\end{subfigure}
\hfill
\begin{subfigure}{0.15\textwidth}
    \includegraphics[width=1.1\textwidth]{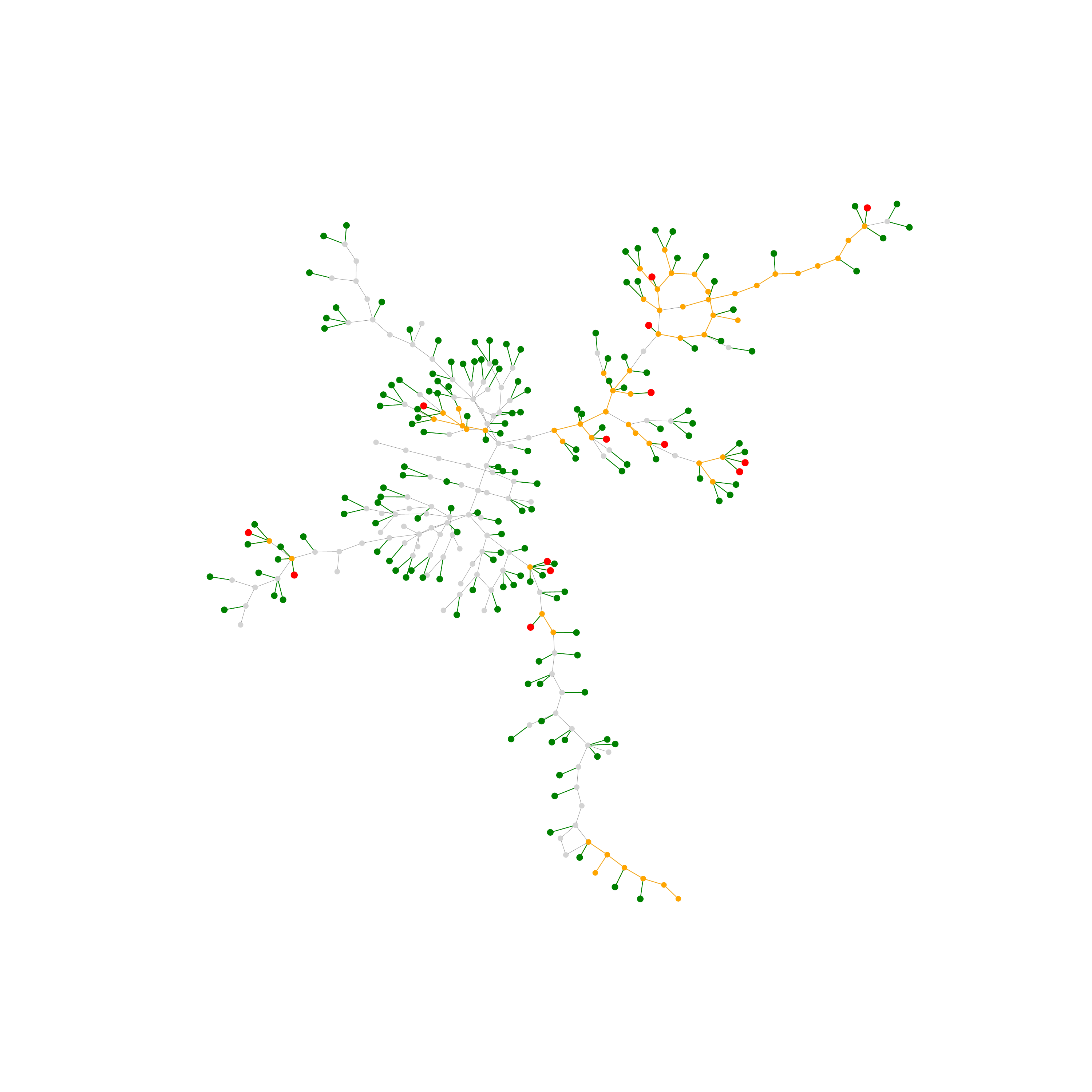}
    \caption{0xd09b 100\%-1}
\end{subfigure}
\\
\begin{subfigure}{0.15\textwidth}
    \includegraphics[width=1.1\textwidth]{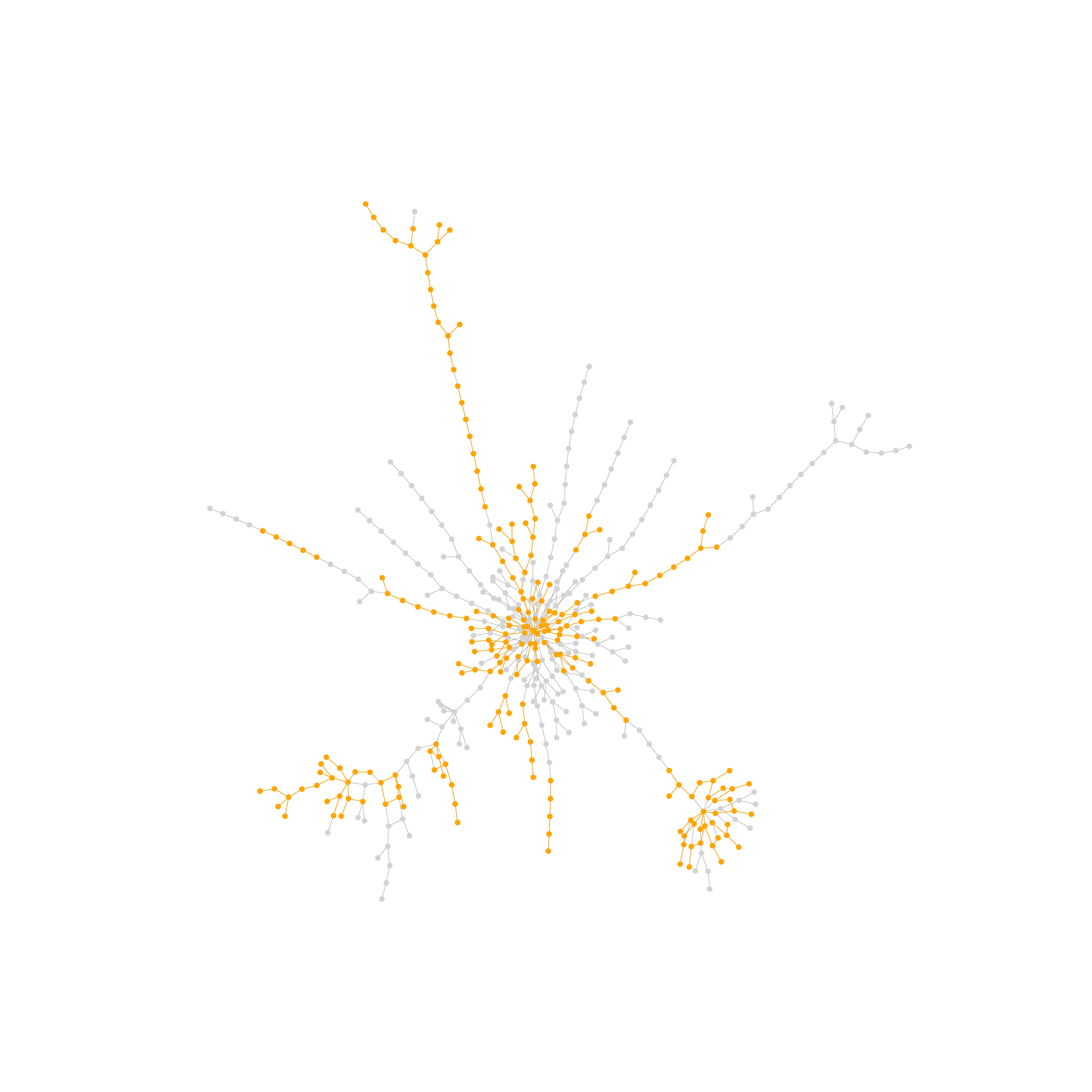}
    \caption{0x8aa5 Clean}
\end{subfigure}
\hfill
\begin{subfigure}{0.15\textwidth}
    \includegraphics[width=1.1\textwidth]{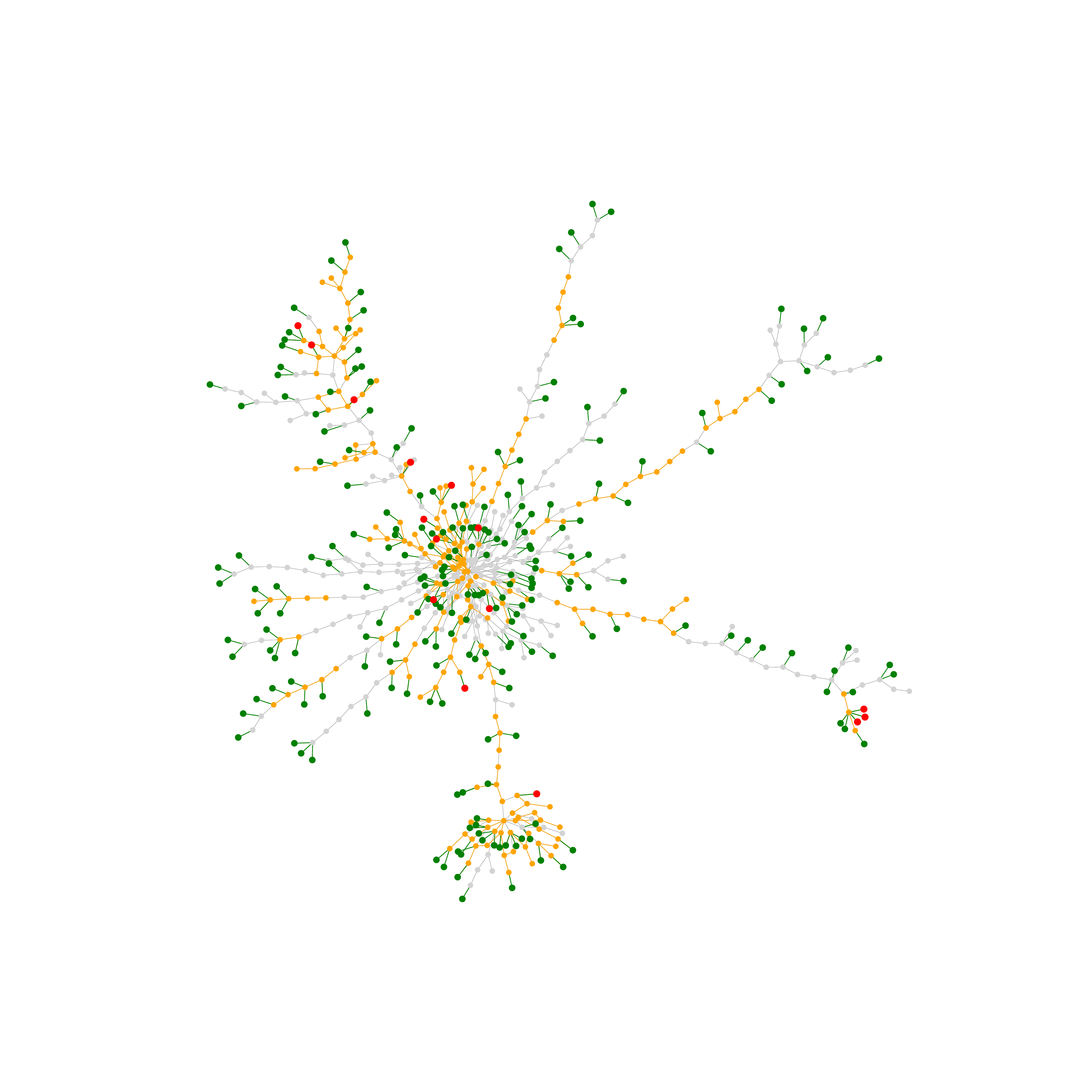}
    \caption{0x8aa5 50\%-1}
\end{subfigure}
\hfill
\begin{subfigure}{0.15\textwidth}
    \includegraphics[width=1.1\textwidth]{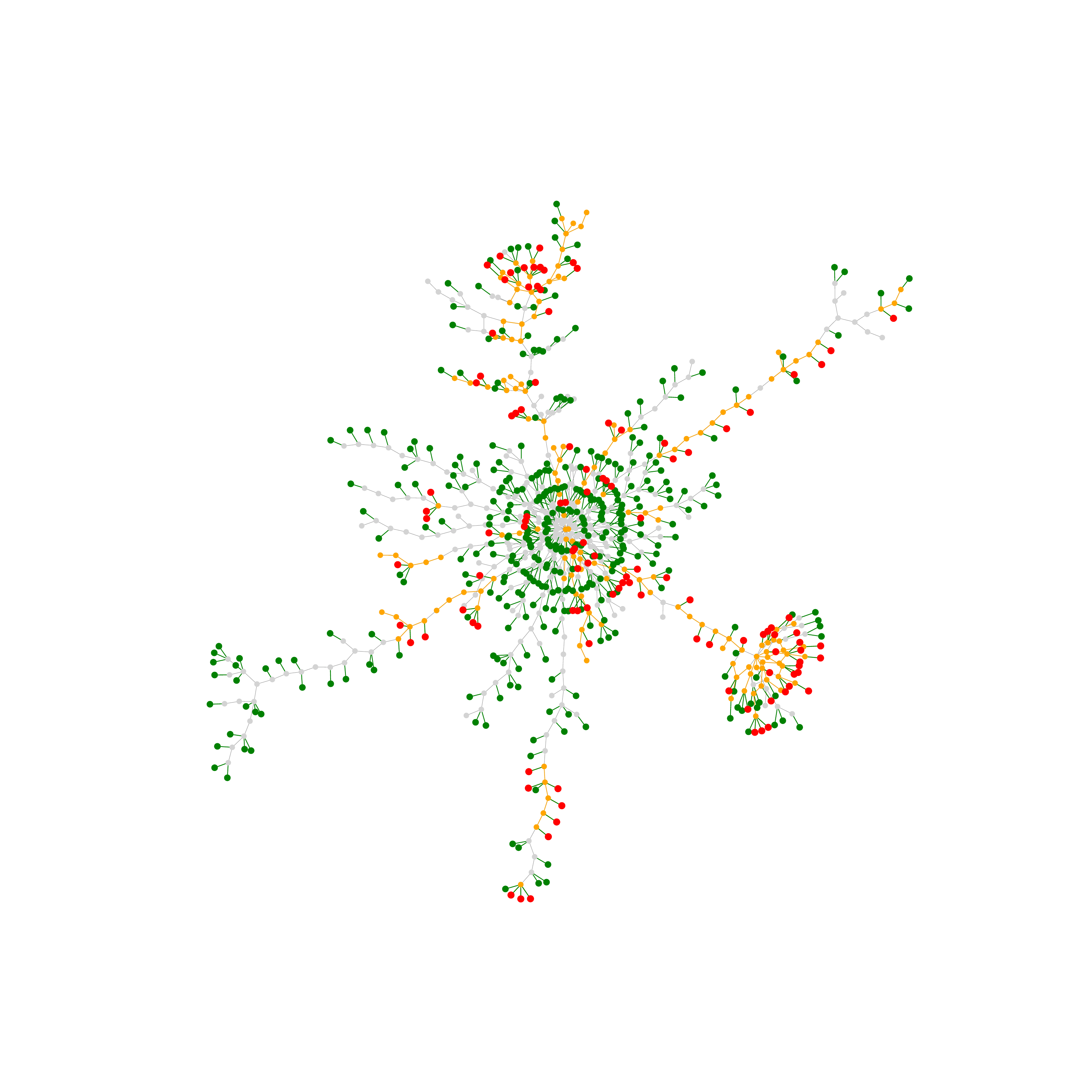}
    \caption{0x8aa5 100\%-1}
\end{subfigure}
\caption{Visualizations of two AECs' explanations.}
\label{fig:visl}
\end{figure}

\begin{figure*}[t]
    \centering
    \begin{subfigure}{0.32\textwidth}
        \centering
        \includegraphics[width=5.5cm, height=3.3cm]{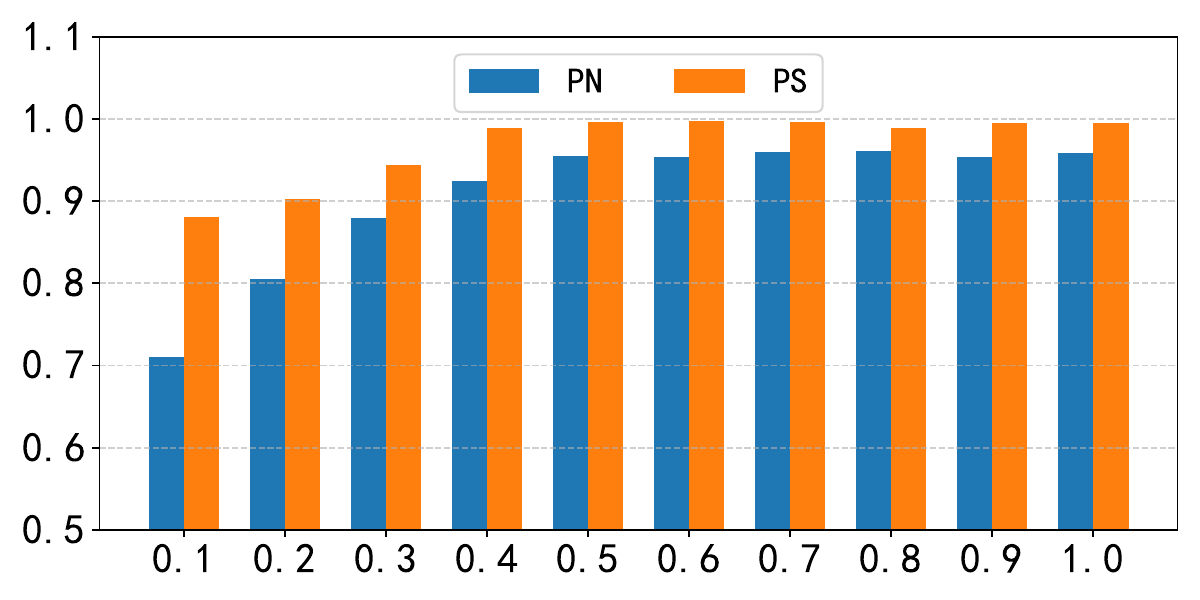}
        \caption{}
        \label{fig:param_alpha}
    \end{subfigure}
    \begin{subfigure}{0.32\textwidth}
        \centering
        \includegraphics[width=5.5cm, height=3.3cm]{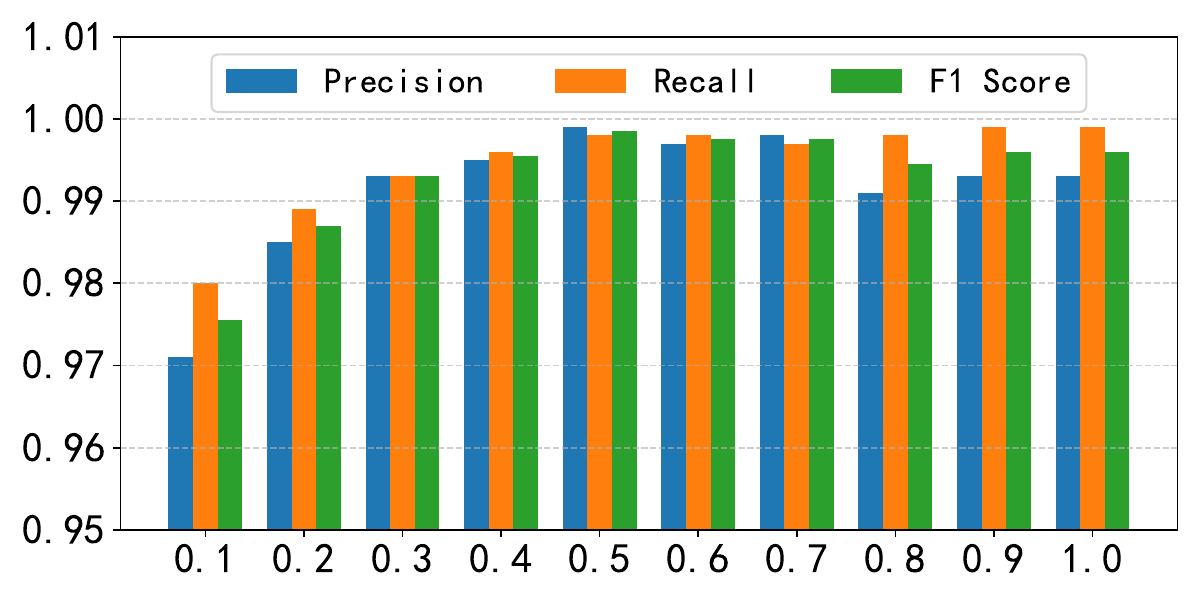}
        \caption{}
        \label{fig:param_beta}
    \end{subfigure}
    \begin{subfigure}{0.32\textwidth}
        \centering
        \includegraphics[width=5.5cm, height=3.2cm]{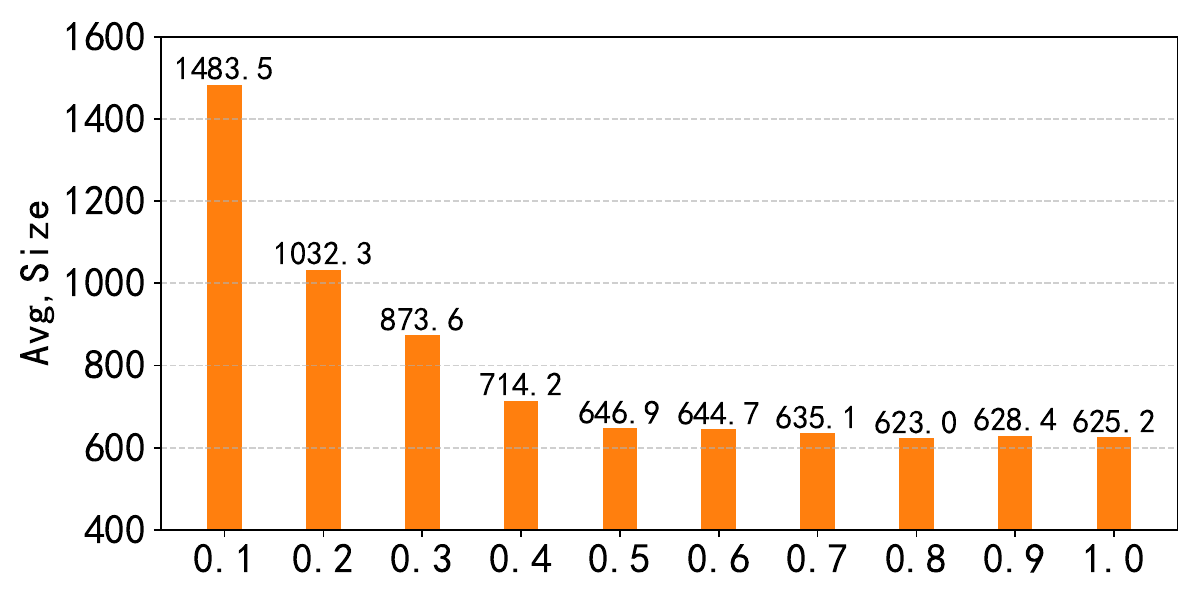}
        \caption{}
        \label{fig:param_gamma}
    \end{subfigure}
    \caption{Parameter sensitive evaluation of: (a) $\alpha$ (b) $\beta$ (c) $\gamma$.}
    \label{fig:param_sens}
\end{figure*}

\subsection{Hyperparameter Sensitive}
We investigate the sensitivity of three key hyperparameters: $\alpha$, $\beta$, and $\gamma$, as shown in Fig.~\ref{fig:param_sens}. Specifically, since $\alpha$ is the weight of $\mathcal{L}_{\mathrm{CF}}$ and directly affects explanation generation, Fig.~\ref{fig:param_alpha} presents the \textbf{PS} and \textbf{PN} values for different $\alpha$ (0.1--1.0). We observe that \textbf{PS} and \textbf{PN} stabilize when $\alpha \geq 0.4$, while smaller $\alpha$ leads to lower values due to insufficient $\mathcal{L}_{\mathrm{CF}}$ weight for accurate explanations. Fig.~\ref{fig:param_beta} shows that $\beta$ (the weight of $\mathcal{L}_{\mathrm{MI}}$) achieves the highest F1 between 0.5 and 0.7, indicating optimal performance. Excessive or insufficient $\beta$ results in lower F1, suggesting that both over- and under-emphasizing $\mathcal{L}_{\mathrm{MI}}$ can hinder detection performance. For $\gamma$ (the weight of $\mathcal{L}_{\mathrm{SP}}$), Fig.~\ref{fig:param_gamma} shows that when $\gamma \geq 0.5$, the average size of the factual subgraph remains below 650, while too small $\gamma$ (e.g., 0.1) leads to overly large and imprecise subgraphs.

\subsection{Runing Time Comparison}

\begin{table}[t]
    \centering
    \small
    \begin{tabular}{ l | c c }
        \toprule
        Methods & \textbf{Training (s)} & \textbf{Test (s)} \\
        \midrule
        GCN & 557.5 & 0.78 \\
        GAT & 641.2 & 0.80 \\
        GIN & 570.6 & 0.74 \\
        RGCN & 714.8 & 0.81 \\
        \midrule
        GNNExlainer & 1,115.1 & 1.33 \\
        PGExlainer & 1,158.0 & 1.12 \\
        CF$^2$ & 1,248.9 & 1.58 \\
        NSEG & 1,223.2 & 1.46 \\
        \midrule
        SEASONED & 1,122.1 & 1.40 \\
        \bottomrule
    \end{tabular}
    \caption{Running time comparison.}
    \label{tab:runing_time}
\end{table}

As shown in Table~\ref{tab:runing_time}, although SEASONED requires more training time compared to GNN-based detectors, it still completes training within 20 minutes, which is within an acceptable range. Compared to post-hoc explainers, SEASONED’s training and inference times are comparable to or slightly better than theirs.

\section{Related Work}
\subsection{Safeguarding DeFi Attacks}
Significant efforts have been dedicated to protecting smart contracts, with research focusing on both the victim and attacker sides.
From the victim side, the primary focus is detecting vulnerabilities. 
A wide range of analysis techniques have been explored, with notable tools including oyente \cite{luu2016making}, 
\textsc{Securify} \cite{tsankov2018securify}, 
Clairvoyance \cite{xue2020cross}, 
Echidna \cite{grieco2020echidna}, 
and \textsc{Sailfish} \cite{bose2022sailfish}.
From the attacker perspective, research efforts focus on detecting transactions and contracts used by attackers.
Among the transaction-specific detection efforts are SODA \cite{chen2020soda},
\textsc{TxSpector} \cite{zhang2020txspector},
DEFIER \cite{su2021evil},
HORUS \cite{ferreira2021eye},
and EthScope \cite{wu2022time}.
Among the contract-specific investigation are Forta \cite{forta2023}, LookAhead \cite{ren2024lookahead}, and Skyeye \cite{wang2024skyeye}.
Researchers have also explored methods like graph neural network \cite{zhuang2020smart,he2024code}, deep reinforcement learning \cite{zhang2022reentrancy}, and transfer learning \cite{sendner2023smarter} to detect various vulnerabilities.
Recent studies have employed generative AI tools for vulnerability detection \cite{sun2023gpt,wang2024smartinv}.


\subsection{Graph Learning and Explanation}
As a powerful paradigm for graph learning, GNNs achieve significant success in various domains, such as social network analysis \cite{hu2023cost}, molecular property prediction \cite{NEURIPS2023_4db8a681}, and recommendation systems \cite{zhang2024hi}. Despite their success, GNNs are often criticized for their lack of interpretability, which hinders their adoption in sensitive domains. To address this, post-hoc explainers have been proposed to generate subgraph-based explanations. For instance, GNNExplainer \cite{gnnexplainer}, PGExplainer \cite{pgexplainer}, and RG-Explainer \cite{rgexplainer} generate subgraph-based explanations by perturbing input data, learning parameterized mappings, or using reinforcement learning, respectively. NSEG \cite{nseg} further seeks necessary and sufficient explanations via optimization. However, these methods are not jointly trained with the detection model and may produce biased or inaccurate explanations. To address these issues, self-explanatory GNNs have been developed. SE-GNN \cite{segnn} and ProtGNN \cite{protgnn} provide built-in explanations by modeling node similarity and prototype matching, with ProtGNN utilizing Monte Carlo tree search for subgraph exploration. SUNNY-GNN \cite{sunnygnn} improves the quality of the explanation through structural perturbations and contrastive learning. Nevertheless, these self-explanatory methods cannot use explanations to enhance detection and may yield suboptimal predictions.

\section{Conclusion}
In this paper, we present SEASONED, an effective, self-explanatory, and robust framework for detecting adversarial exploiter contracts. SEASONED leverages a novel semantic relation graph to capture semantical information and employs a self-counterfactual explainable detector for both detection and core logic explanation. Furthermore, SEASONED enhances detection robustness, generalizability, and data efficiency by extracting the representative information from explanations. Experiments show that SEASONED outperforms state-of-the-art methods in detection performance, adversarial robustness, generalization to new appeared contracts, and data efficiency. As an innovative integration of explainable graph learning and AEC detection, SEASONED advances DeFi security.




%

\bibliographystyle{IEEEtran}
\bibliography{main_ndss}
\end{document}